\documentclass[%
reprint,
 prd,
superscriptaddress,
nofootinbib,
aps,
dvipsnames
]{revtex4-2}
\usepackage[utf8]{inputenc}

\usepackage{amsmath, amsthm, amsfonts,  amssymb, latexsym}
\usepackage{graphicx}
\usepackage{dcolumn}
\usepackage{bm}
\usepackage{hyperref}
\usepackage[mathlines]{lineno}
\usepackage{orcidlink}
\usepackage{breakurl}

\usepackage[normalem]{ulem}
\usepackage{xcolor}

\newcommand{\bs}[1]{\boldsymbol{#1}}
\usepackage{enumerate}
\usepackage{mathrsfs}
\usepackage{graphicx}
\usepackage{ragged2e}
\usepackage{subcaption}
\usepackage{xfrac}
\usepackage{lmodern}
\usepackage{multirow}
\usepackage[toc,page]{appendix}

\usepackage{tikz}
\usetikzlibrary{arrows.meta}
\usetikzlibrary{decorations.pathmorphing}
\usetikzlibrary{decorations.markings}
\usetikzlibrary{angles,
                quotes}

\newcommand{\ud}{\mathrm{d}}

\newtheorem{proposition}{Proposition}[section]
\newtheorem{theorem}[proposition]{Theorem}
\newtheorem{lemma}[proposition]{Lemma}

\newtheorem{corollary}[proposition]{Corollary}
\newtheorem{definition}[proposition]{Definition}

\newcommand{\sM}{{\mathscr M}}
\newcommand{\cM}{{\mathscr M}}
\newcommand{\cN}{{\mathscr N}}
\newcommand{\cC}{{\mathscr C}}

\newcommand{\sI}{{\mathscr I}}
\newcommand{\sH}{{\mathscr H}}
\newcommand{\cH}{{\mathscr H}^+}
\newcommand{\sB}{{\mathscr B}}
\newcommand{\sL}{{\mathscr L}}

\newcommand{\sD}{{\mathscr D}}

\newcommand{\cU}{{\mathscr U}}

\newcommand{\be}{\begin{equation}}
\newcommand{\ee}{\end{equation}}
\newcommand{\bea}{\begin{eqnarray}}
\newcommand{\eea}{\end{eqnarray}}

\begin{document}
	
	\title{New Horizons in Effective Field Theory?}

		\author{Stefan Hollands \orcidlink{0000-0001-6627-2808}}
	\email[Contact author: ]{stefan.hollands@uni-leipzig.de}
	\affiliation{%
	Institute of Theoretical Physics, Leipzig University, Br{\"u}derstra{\ss}e 16, 04103 Leipzig, Germany
}%
\affiliation{%
	Max Planck Institute for Mathematics in Sciences (MiS), Inselstra{\ss}e 22, 04103
	Leipzig, Germany
}

\author{Dustin Urbiks}
	\affiliation{%
	Institute of Theoretical Physics, Leipzig University, Br{\"u}derstra{\ss}e 16, 04103 Leipzig, Germany
}%

	\begin{abstract}
We consider the most general parity symmetric effective scalar tensor theory in four dimensions containing terms up to fourth derivative order in the Lagrangian. It has been shown [H.S. Reall, Phys. Rev. D 103 (2021), 084027] that this theory has three polarizations generically goverened by different (nested) propagation cones, neither of which in general coincides with the lightcone as defined by the metric. Consequently, the notion of black hole horizon must be defined relative to the widest propagation cone, and not with respect to the metric. We provide two theorems stating that, nevertheless, the horizon of a 
\emph{stationary} black hole is null with respect to the metric, and that, in fact, all three propagation cones touch on the horizon. The conditions in these theorems allow for rotating black holes. Our theorems thereby suggest that the notion of Killing horizon, central in most discussions of black hole thermodynamics, retains its fundamental status, and that certain thermodynamic paradoxes associated with multiple propagation cones are evaded.
	\end{abstract}

	\maketitle

\section{Introduction}\label{sec:intro}

In relativisitic physics, the propagation of disturbances as described by, e.g., the scalar wave equation $g^{ab}\nabla_a \nabla_b \Phi=0$
is well-known to be causal. That means that the retarded respectively advanced solutions of the corresponding equation with a delta-function source located at a point $x \in \mathscr{M}$ in a causally well-behaved spacetime\footnote{E.g., the spacetime should be globally hyperbolic, see e.g., \cite{Waldbook}.} $(\sM,g_{ab})$ vanish outside the causal future $J^+(x)$ respectively past $J^-(x)$ of this point. The causal future respectively past consist of all events that can be connected to $x$ by a causal (i.e., timelike or null) past respectively future directed curve. Since a curve is timelike or null if its tangent vector $u^a$ satisfies\footnote{Our signature convention is $(-+++)$.} $g_{ab} u^a u^b \le 0$, the notion of causality for the scalar field propagation is tied to the lightcones of the metric, $g_{ab}$. An equivalent formulation is that any solution to $g^{ab}\nabla_a \nabla_b \Phi=0$ is uniquely determined by its initial data on some subset $\cU \subset \Sigma$ of a Cauchy-surface $\Sigma$ within the domain of dependence, $D(\cU)$, of $\cU$. Again, $D(\cU)$ is determined by the metric as the set of all points such that any inextendible future or past directed causal curve must intersect $\cU$.
For modern accounts of these well-known facts see, e.g., \cite{Ringstrom}.

Up to issues of gauge, i.e., after passing to suitable gauge equivalence classes or after suitably fixing the gauge, the same concept of causality applies to solutions $A_a$ of the Maxwell equations, or solutions $h_{ab}$ of the linearized vacuum Einstein equations. In fact, stating causality in terms of the domain of dependence, the same holds true even at the non-linear level for the coupled Einstein-Maxwell-scalar wave equation system, see, e.g.,  \cite{Ringstrom}. Given that information is thought of as being transported by waves, it is thereby justified to define the notion of  black hole region in such a theory as $\mathscr{M} \setminus J^-(\mathscr{I}^+)$, where $\mathscr{I}^+$ is future null infinity. The future horizon thus is 
\begin{equation}
\label{futurehor}
\mathscr{H}^+ = \partial[\mathscr{M} \setminus J^-(\mathscr{I}^+)].
\end{equation}
Since the notion of causal past, $J^-$, depends on the metric $g_{ab}$, so obviously does the notion of a black hole. 

In alternative theories of gravity, there may be several different metrics, or even if there is only a single metric, various ``polarizations'' of the fields of the theory may propagate according to a different ``effective'' metric due to the structure of the highest derivative terms in the equations of motion. As a simple minded toy model of such situations, one may consider two scalar fields $\Phi_A, \Phi_B$ obeying 
\begin{equation}\label{ABwaveeq}
    \begin{split}
        &g_A^{ab} \nabla_a \nabla_b \Phi_A - V_{,\Phi_A}(\Phi_A,\Phi_B) = 0 \, ,\\
        &g_B^{ab} \nabla_a \nabla_b \Phi_B - V_{,\Phi_B}(\Phi_A,\Phi_B) = 0 \, ,
    \end{split}
\end{equation}
where $g_{Aab}, g_{Bab}$ are different metrics and $V$ is a potential
coupling the two fields. In such a theory, the notion of causal past and future should obviously be defined w.r.t. the curves whose tangent is causal w.r.t. at least one--but possibly not the other--metric, i.e., the ``widest'' lightcone at every point. 


For the sake of illustration, envisage a hypothetical black hole spacetime as in fig. \ref{fig:ABhorizons} with nested horizons $\mathscr{H}_A$ and $\mathscr{H}_B$ defined relative to the $A$ 
and $B$ metrics. We assume that there is a vector field $\chi^a$ which Lie-derives \emph{both} metrics, and
that both metrics are asymptotically flat, with $\chi^a$ timelike near null infinity (simultaneously for both metrics). We assume further that the $B$ metric has the widest lightcones.

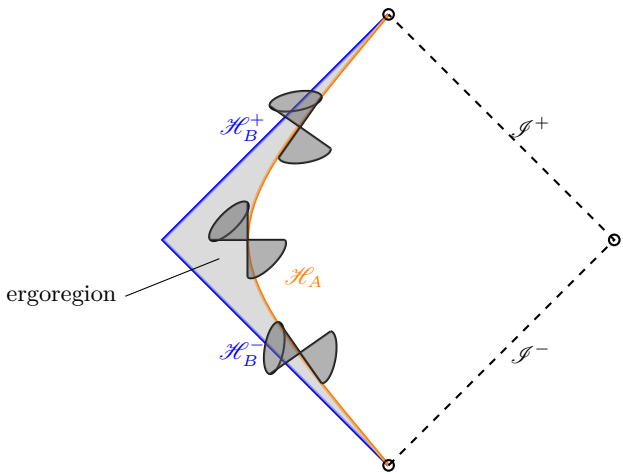
\begin{figure}[h]
\centering
\begin{tikzpicture}[scale=0.5]
    \draw[very thick, blue] (0,0) -- (-6,6) node[midway, below, left]{$\mathscr H_B^-$} -- (0,12) node[midway, above, left]{$\mathscr H_B^+$} ;
    \draw[very thick, orange] (0,0) .. controls (-5,6) .. (0,12) ;
    \draw[orange] (-3, 5) node[right]{$\mathscr H_A$} ;
    \draw[thick, dashed] (0,0) -- (6,6) node[midway, below, right]{$\mathscr I^-$} -- (0,12) node[midway, above, right]{$\mathscr I^+$} ;
    \draw[thick] (0,0) circle(4pt) (6,6) circle(4pt) (0,12) circle(4pt) ;
    \fill[black!20, fill opacity=0.7] (0,12) -- (-6,6) -- (0,0) .. controls (-5,6) .. (0,12) ;
    \draw (-4.5,5.5) -- (-7,4.5) node[left]{ergoregion} ;
    \begin{scope}[shift={(-2.35,3)}, rotate=80]
        \draw[thick, black!90] (-0.7,-0.7) -- (0.7,0.7) ;
        \draw[thick, black!90] (0.7,-0.7) -- (-0.7,0.7) ;
        \draw[thick, black!90] (0.7, 0.7) arc[x radius=0.7, y radius=0.25, start angle=0, end angle=360] ;
        \draw[thick, black!90] (0.7, -0.7) arc[x radius=0.7, y radius=0.25, start angle=0, end angle=-180] ;
        \fill[black!60, fill opacity=0.5] (0,0) -- (0.7,0.7) arc[x radius=0.7, y radius=0.25, start angle=0, end angle=180] -- (0,0) ;
        \fill[black!60, fill opacity=0.5] (0,0) -- (0.7,-0.7) arc[x radius=0.7, y radius=0.25, start angle=0, end angle=-180] -- (0,0) ;
    \end{scope}
    \begin{scope}[shift={(-3.75,6)}, rotate=45]
        \draw[thick, black!90] (-0.7,-0.7) -- (0.7,0.7) ;
        \draw[thick, black!90] (0.7,-0.7) -- (-0.7,0.7) ;
        \draw[thick, black!90] (0.7, 0.7) arc[x radius=0.7, y radius=0.25, start angle=0, end angle=360] ;
        \draw[thick, black!90] (0.7, -0.7) arc[x radius=0.7, y radius=0.25, start angle=0, end angle=-180] ;
        \fill[black!60, fill opacity=0.5] (0,0) -- (0.7,0.7) arc[x radius=0.7, y radius=0.25, start angle=0, end angle=180] -- (0,0) ;
        \fill[black!60, fill opacity=0.5] (0,0) -- (0.7,-0.7) arc[x radius=0.7, y radius=0.25, start angle=0, end angle=-180] -- (0,0) ;
    \end{scope}
    \begin{scope}[shift={(-2.35,9)}, rotate=10]
        \draw[thick, black!90] (-0.7,-0.7) -- (0.7,0.7) ;
        \draw[thick, black!90] (0.7,-0.7) -- (-0.7,0.7) ;
        \draw[thick, black!90] (0.7, 0.7) arc[x radius=0.7, y radius=0.25, start angle=0, end angle=360] ;
        \draw[thick, black!90] (0.7, -0.7) arc[x radius=0.7, y radius=0.25, start angle=0, end angle=-180] ;
        \fill[black!60, fill opacity=0.5] (0,0) -- (0.7,0.7) arc[x radius=0.7, y radius=0.25, start angle=0, end angle=180] -- (0,0) ;
        \fill[black!60, fill opacity=0.5] (0,0) -- (0.7,-0.7) arc[x radius=0.7, y radius=0.25, start angle=0, end angle=-180] -- (0,0) ;
    \end{scope}
\end{tikzpicture}
\caption{\justifying A hypothetical spacetime containing two nested horizons $\mathscr H_A$ and $\mathscr H_B$ defined w.r.t. $g_{Aab}$ and $g_{Bab}$, respectively. Since $\chi^a$
is a "boostlike" Killing field, it does not vanish on $\mathscr H_A$, meaning that this set must be \emph{either} a future \emph{or}
a past horizon w.r.t. $g_{Aab}$. In particular, this metric cannot have a "time reflection" symmetry, indicating a fundamental time-asymmetry of the setup.
} 
\label{fig:ABhorizons}
\end{figure}

Even though in our hypothetical spacetime, the inner horizon, $\mathscr{H}_B$, is thus of course the true black hole horizon, the nested structure of $\mathscr{H}_A$ and $\mathscr{H}_B$  
may give rise to interesting effects and even paradoxes, as has been noted by several authors \cite{Dubovsky:2006vk, Eling:2007qd, Jacobson:2010fat}, see also \cite{Benkel:2018abt} for further discussion and references. 

A first observation, which by itself does not raise a paradox, is that if the (outer) horizon $\mathscr{H}_A$ is a non-degenerate Killing horizon w.r.t. $\chi^a$, i.e., $\mathscr{H}_A$ is not rotating relative to infinity and $\chi^a$ is tangent and normal to $\mathscr{H}_A$ w.r.t. the $A$ metric, then the horizon Killing field $\chi^a$ must be spacelike w.r.t. the $A$ metric somewhere between the $A$ and $B$ horizons\footnote{This immediately follows, e.g., from Eq. $g_A^{ac}\nabla_c (\chi_b\chi^b) = -2\kappa_A \chi^a$ and the fact that $\kappa_A >0$, see e.g., \cite[sec. 12.5]{Waldbook}.}. Thus, the region between the $A$ and $B$ horizons includes a kind of ergoregion for $A$ particles, which enables an analogoue of the Penrose process for extracting energy from black holes; see e.g., \cite{Waldbook} for an explanation of the standard version of this process in the Kerr black hole ergoregion.

More concretely, in a mechanical version of this process, consider 
an $A$ and a $B$ particle with momenta $p_{Aa}$ and $p_{Ba}$, moving along lightlike geodesics in the $A$ respectively $B$ metrics. They collide inside the region between the $A$ and $B$ horizons at a point where 
$\chi^a$ is spacelike w.r.t. the $A$ metric. 
After the collision, the particles have lightlike momenta $p_{Aa}'$ and $p_{Ba}'$ relative to the respective metrics. 
One can arrange that $E'_A = -\chi^a p_{Aa}'<0$, and that the outgoing 
$B$ particle can re-exit the $A$ horizon, while the outgoing $A$ particle necessarily falls into the black hole, see fig. \ref{fig:collision}. Obviously, 
by energy conservation\footnote{This follows from the fact that $\chi^a$ Lie-derives both the $A$ and $B$ metric.} $E_A+E_B=E_A'+E_B'$ and $E_A>0$ since $\chi^a$ is timelike near infinity, so $E_B'>E_B$. For a closely related setting see, e.g., \cite{Eling:2007qd}.

\begin{figure}[h]
\centering
\begin{tikzpicture}[scale=0.4]
    \draw[thick, blue] (0,0) circle(2) ; 
    \draw[thick, blue] (2,0) node[right]{$\mathscr H_B$} ;
    \fill[black!30, fill opacity=0.7] (0,0) circle(2) ;
    \draw[thick, orange] (0,0) circle(5.5) ;
    \draw[thick, orange] (5.5,0) node[right]{$\mathscr H_A$} ;

    \draw[thick, postaction={decorate, decoration={markings, mark=at position 0.5 with {\arrow{Latex[length=3mm, width=2.5mm]}}}}] (-9,-3) -- node[midway, above]{$p_A$} (-3,0) ;
    \draw[thick, postaction={decorate, decoration={markings, mark=at position 0.35 with {\arrow{Latex[length=3mm, width=2.5mm]}}}}] (-3,0) .. controls (-4,3) .. (0,2) node[midway, left]{$p'_A$} ;
    \draw[thick, postaction={decorate, decoration={markings, mark=at position 0.5 with {\arrow{Latex[length=3mm, width=2.5mm]}}}}] (-1,-9) -- (-3,0) node[midway, right]{$p_B$} ;
    \draw[thick, postaction={decorate, decoration={markings, mark=at position 0.5 with {\arrow{Latex[length=3mm, width=2.5mm]}}}}] (-3,0) -- (1,9) node[midway, right]{$p'_B$} ;

    \fill[black] (-3,0) circle(4pt) ;
\end{tikzpicture}
\caption{\justifying The collision between two particles with 4-momenta $p_{Aa}$ and $p_{Ba}$.}
\label{fig:collision}
\end{figure}
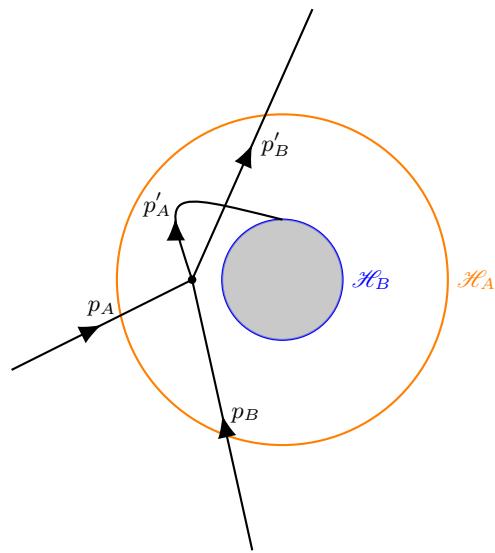

It is not inconceivable that there could be a field-theoretic analog of this setup choosing high-frequency Gaussian beam solutions $\Phi_A$
and $\Phi_B$ whose initial support is close to the trajectories co-tangent to $p_{Aa}$ and $p_{Ba}$. Indeed, before and after the collision, the potential terms in the coupled wave equations \eqref{ABwaveeq} can be neglected and the solutions are approximately equal to free Gaussian beams, as described, e.g., in \cite{Green:2015kur}, and matching these solutions across the collision might be possible.  

While interesting and potentially observable via an analogue of the superradiance effect \cite{Starobinski}, energy extraction from a black hole by itself does not raise any physical paradox. However, as argued in \cite{Dubovsky:2006vk, Eling:2007qd, Jacobson:2010fat}, by exploiting the hypothetical situation with nested $A$ and $B$ horizons further, one may enable perpetual motion creating tensions with the generally accepted laws of thermodynamics. 

Furthermore, while it is not totally clear how well-motivated one should consider theories with several metrics, different \emph{effective metrics} frequently arise in effective field theory (EFT) type generalizations of Einstein-matter theories with higher derivative terms in the action and/or theories with unconventional fields such as Einstein-aether theories, see e.g., \cite{Adam:2021vsk}. Although there is only a single dynamical metric $g_{ab}$ in such theories, different ``polarizations'' of the coupled field theoretic content are governed by different effective metrics.\footnote{As has been shown by \cite{Reall:2021voz}, the effective propagation cones may not even be defined by a quadratic equation arising from an effective metric but by higher order, e.g. quartic, polynomials arising from the principal symbol of the highest order operator in the Euler-Lagrange equations. We refer to secs. \ref{sec:propcones}, \ref{sec:slowness} for a short review of some of his results.} 

While the notion of propagation cones in EFTs is thus in general different from the standard notion of lightcone w.r.t. $g_{ab}$, it may of course happen that for {\it stationary}, regular, asymptotically flat black hole solutions, all those notions precisely coincide \emph{at the event horizon}. Evidence that this may be the case has been provided by \cite{Benkel:2018qmh, Tanahashi:2017kgn,Izumi:2014loa,Ali:2025ybt}. 
These authors study various scalar-tensor theories in the Horndeski class \cite{Benkel:2018qmh, Tanahashi:2017kgn} or higher dimensions \cite{Izumi:2014loa} with second order equations of motion, or various EFTs for just a metric \cite{Ali:2025ybt}. They demonstrate at various levels of generality that a Killing horizon relative to $g_{ab}$ is automatically a characteristic surface of a suitably defined propagation cone of the theory. Since the characteristic directions determine the propagation of discontinuities \cite{couranthilbert}, hence the boundary of the domain of dependence of the equation, it follows that a Killing horizon can be the causal edge just as in standard Einstein-scalar theories.

While these results provide an interesting link between Killing- and propagation cone horizons, the fundamental question is whether, vice versa, any propagation cone horizon is also a Killing horizon. If true, this would eliminate any of the concerns \cite{Dubovsky:2006vk, Eling:2007qd, Jacobson:2010fat} with perpetual motion. It would also reinforce the usual laws of black hole mechanics in alternative- or effective gravity theories with higher derivatives and/or additional unconventional fields, since these laws are usually formulated with regard to Killing horizons \cite{Iyer:1994ys,Iyer:1995kg,Hollands:2024vbe,Hollands:2022fkn}. 

Positive evidence for this scenario has been given by \cite{Benkel:2018qmh} for spherically symmetric static black hole solutions in the Horndeski class. In the present paper, we investigate \emph{rotating} black holes in the class of scalar-tensor EFTs in dimension $d=4$ described by the Lagrange density
\begin{equation}
\begin{split}
    \bs{L} =& \, \, [-V(\Phi)+R+X] \bs{\epsilon} \\ 
    &+ \frac{1}{2} \ell^2 \alpha(\Phi) X^2 \bs{\epsilon} + \frac{1}{16} \ell^2 \beta(\Phi) \bs{L}_{\rm GB}, 
    \label{lagrangian}
\end{split}
\end{equation}
where $X=-\frac{1}{2}g^{ab}\nabla_a\Phi\nabla_b\Phi$ is the kinetic term of a scalar field $\Phi$, where $\bs{L}_{\rm GB}$ is the Gauss-Bonnet
density, and where $\alpha,\beta,V$ are, in principle, free real functions. The above Lagrangian is a generalization of the usual Einstein-scalar field theory---the first line---containing additional terms with four derivatives in the second line. The overall size of the four-derivative couplings is dictated by the constant $\ell$ which has the dimension of a length. It has been shown \cite{Weinberg:2008hq} that this theory in the Horndeski class is, up to field redefinitions, the most general scalar-tensor EFT with up to four derivative terms in the action under the assumption of parity symmetry, thus it represents a quite general class of models. 

The Euler-Lagrange equations of \eqref{lagrangian} contain terms with at most two derivatives acting on each occurrence of $(\Phi,g_{ab})$ \cite{Reall:2021voz}. The nature of the propagation cones of these equations has been analyzed in detail by \cite{Reall:2021voz} (we will recall some of these results in sec. \ref{sec:propcones}). In particular, it has been shown \cite{Reall:2021voz} that the theory has, at least at the linearized level, three distinct ``polarizations'' with distinct propagation cones. The covectors in these cones, which are co-tangent to the characteristics, are defined as solutions to either a quadratic or quartic equation, with the inner sheet of the quartic cone determining the ``fastest'' polarization, see fig. \ref{fig:ccs}. Thus, the characteristics of the fastest polarization determine the physical notion of black hole horizon, $\mathscr{H}^+$. 

In the present paper, we provide two theorems showing that this horizon must in fact be a Killing horizon. Both theorems assume that the black hole solution $(\Phi,g_{ab})$ is real analytic, stationary (possibly rotating) and given to us as a real analytic family in the EFT parameter $\ell$. Furthermore, it is assumed that the $\ell=0$ solution 
of the usual Einstein-scalar theory is a stationary black hole spacetime with non-degenerate horizon.  Theorem A also assumes that the generators of the characteristic flow on $\mathscr{H}^+$---the propagation cone analogue of null-geodesics---are collinear with a Killing vector field, $\chi^a$, at all points of $\mathscr{H}^+$, and that $\beta'$ does not vanish. On the other hand, Theorem B assumes that the metric has a ``$t$-$\phi$'' reflection symmetry analogous to Kerr, 
and that $V \ge 0$. The assumptions are discussed in detail in sec. \ref{sec:thmA} (Theorem A) respectively sec. \ref{sec:thmB} (Theorem B). 

Both theorems conclude that $\mathscr{H}^+$ is actually a Killing horizon w.r.t. $g_{ab}$ with constant surface gravity for sufficiently small $\ell$. At this horizon, the three-sheet structure of the various propagation cones degenerates, as shown in fig. \ref{fig:1}.

\begin{figure}[h]
\centering
\begin{tikzpicture}[scale=0.17
]
    
    \draw[very thick] (-26,6) -- (-2,14) ;
    \draw[very thick] (-18,-6) -- (9,3) node[below] {$\cH$} ;
    \draw[thick, teal, -Latex] (-3,-1) -- (-14, 10) ;
    \draw[thick, teal, -Latex] (-6,-2) -- (-17,9) node[above, teal] {$\dot{x}^a$} ;
    \draw[thick, teal, -Latex] (3,1) -- (-8,12) ;
    \draw[thick, teal, -Latex] (6,2) -- (-5,13) ;

    \draw[very thick, orange, fill=orange!50, fill opacity=0.5] (9,9) arc[start angle =0, end angle=180, x radius=9, y radius=1.8] ;
    \fill[orange!60, fill opacity=0.5] (9,9) arc[start angle=0, end angle=-180, x radius=9, y radius=1.8] ;
    \draw[thick, blue] (0,0) -- (1,9) ;
    \fill[fill=blue!50, fill opacity=0.5] (-9,9) -- (0,0) -- (1,9) arc[start angle=0, end angle=-180, x radius=5, y radius=1.4] ;
    \draw[thick, blue, fill=blue!50,  fill opacity=0.5] (1,9) arc[start angle=0, end angle=180, x radius=5, y radius=1.4] ;
    \draw[thick, blue, fill=blue!60, fill opacity=0.5] (1,9) arc[start angle=0, end angle=-180, x radius=5, y radius=1.4] ;
    \draw[thick, red] (0,0) -- (-2,9) ;
    \fill[fill=red!50, fill opacity=0.5] (-9,9) -- (0,0) -- (-2,9) arc[start angle=0, end angle=-180, x radius=3.5, y radius=1] ;
    \draw[thick, red, fill=red!50, fill opacity=0.5] (-2,9) arc[start angle=0, end angle=180, x radius=3.5, y radius=1] ;
    \draw[thick, red, fill=red!60, fill opacity=0.5] (-2,9) arc[start angle=0, end angle=-180, x radius=3.5, y radius=1] ;
    \draw[very thick, orange, fill=orange!50, fill opacity=0.7] (-9,9) -- (0,0) -- (9,9) arc[start angle=0, end angle=-180, x radius=9, y radius=1.8] ;
    \fill (0,0) circle(10pt) ;

    \draw[very thick, -{Latex[scale=1.5]}] (0,0) -- (-11,11) node[above] {$\chi^a$} ;

    \begin{scope}[shift={(-15,-5)}]
    \draw[very thick, orange, fill=orange!50, fill opacity=0.5] (9,9) arc[start angle =0, end angle=180, x radius=9, y radius=1.8] ;
    \fill[orange!60, fill opacity=0.5] (9,9) arc[start angle=0, end angle=-180, x radius=9, y radius=1.8] ;
    \draw[thick, blue] (0,0) -- (1,9) ;
    \fill[fill=blue!50, fill opacity=0.5] (-9,9) -- (0,0) -- (1,9) arc[start angle=0, end angle=-180, x radius=5, y radius=1.4] ;
    \draw[thick, blue, fill=blue!50,  fill opacity=0.5] (1,9) arc[start angle=0, end angle=180, x radius=5, y radius=1.4] ;
    \draw[thick, blue, fill=blue!60, fill opacity=0.5] (1,9) arc[start angle=0, end angle=-180, x radius=5, y radius=1.4] ;
    \draw[thick, red] (0,0) -- (-2,9) ;
    \fill[fill=red!50, fill opacity=0.5] (-9,9) -- (0,0) -- (-2,9) arc[start angle=0, end angle=-180, x radius=3.5, y radius=1] ;
    \draw[thick, red, fill=red!50, fill opacity=0.5] (-2,9) arc[start angle=0, end angle=180, x radius=3.5, y radius=1] ;
    \draw[thick, red, fill=red!60, fill opacity=0.5] (-2,9) arc[start angle=0, end angle=-180, x radius=3.5, y radius=1] ;
    \draw[very thick, orange, fill=orange!50, fill opacity=0.7] (-9,9) -- (0,0) -- (9,9) arc[start angle=0, end angle=-180, x radius=9, y radius=1.8] ;
    \fill (0,0) circle(10pt) ;
    \end{scope}
\end{tikzpicture}
\caption{\justifying The red, orange and blue cones indicate different propagation cones of the theory \eqref{lagrangian} in the tangent space (quartic and quadratic cones). 
At the horizon $\mathscr{H}^+$, these touch at the horizon Killing vector field $\chi^a$ pointing along $\mathscr{H}^+$. At these points, $\chi^a$ is null w.r.t. $g_{ab}$ and satisfies 
$\chi^a \nabla_a \chi^b = \kappa \chi^b$ with constant $\kappa > 0$. $\chi^a$ is also co-linear with the tangent $\dot x^a$ of the null-geodesics ruling $\mathscr{H}^+$. This illustration is exaggerated because all cones should be close to each other in a weakly coupled (small $|\ell|$) theory.}
\label{fig:1}
\end{figure}
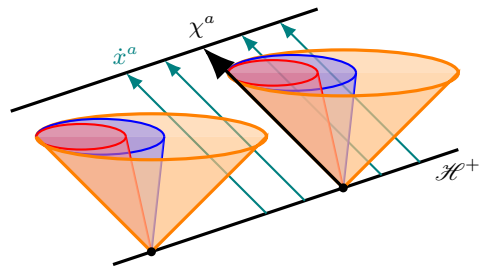

Due to the technical assumptions underlying our theorems, these are not totally definite results, but we believe that they provide evidence of a censorship mechanism which forces stationary propagation cone horizons to be actual Killing horizons, thereby saving the laws of black hole thermodynamics from the paradoxes raised in \cite{Dubovsky:2006vk, Eling:2007qd, Jacobson:2010fat}.

This paper is organized as follows. In sec. \ref{sec:setup} we review 
our EFT and its various propagation cones. In secs. \ref{sec:thmA} 
and \ref{sec:thmB} we present our theorems. Some lengthy technical arguments and formulas needed for the proofs are deferred to various appendices. Our notations and conventions are the same as in \cite{Waldbook}.

\section{Setup}
\label{sec:setup}
\subsection{Equations and Principal Symbols of EFT}

We consider configurations $(\cM,\Phi,g_{ab})$ where $(\cM,g_{ab})$ is a 4-dimensional spacetime and $\Phi$ a real-valued scalar field, which satisfy second order equations of motion (EoM's) produced by the Lagrangian \eqref{lagrangian} of our EFT.  
Written out more fully, this Lagrangian is
\begin{equation}
\begin{split}
    &\bs{L} = (X-V+R) \bs{\epsilon} +\\ 
    & \left(\frac{1}{2} \ell^2 \alpha X^2 + \frac{1}{16} \ell^2 \beta \delta^{a_1a_2a_3a_4}_{b_1b_2b_3b_4} R_{a_1a_2}{}^{b_1b_2}R_{a_3a_4}{}^{b_3b_4} \right)\bs{\epsilon} \, .
    \label{lagrangian1}
\end{split}
\end{equation}
$V(\Phi),\alpha(\Phi),\beta(\Phi)$ are real analytic functions of $\Phi$. 
$X=-\frac{1}{2}g^{ab}\nabla_a\Phi\nabla_b\Phi$ is a shorthand for the kinetic term of $\Phi$. $\delta^{a_1a_2a_3a_4}_{b_1b_2b_3b_4}/4!$ is the projector onto the totally anti-symmetric part of a rank four tensor. The constant $\ell$ governs the strength of the higher derivative terms in the second line and has the dimension of a length.

The EoM's obtained from varying \eqref{lagrangian} w.r.t. $g_{ab}$ and $\Phi$ are given by $E_g^{ab}=E_\Phi=0$, where \cite{Reall:2021voz}
\begin{align}
\label{EoMg}
    E_g^{ab} &= G^{ab} - \left( \frac{1}{2} + \alpha X \right) \nabla^a\Phi \nabla^b \Phi - \frac{1}{2} \left( V-X + \frac{1}{2} \ell^2\alpha X \right) \notag \\
    &-\frac{1}{4} \ell^2 \epsilon^{ac_1c_2c_3}\epsilon^{bd_1d_2d_3} R_{c_1c_2d_1d_2} (\nabla^2\beta)_{c_3d_3} \, ,
\end{align}
($G^{ab}$ is the Einstein tensor and indices are raised with $g^{ab}$ throughout this paper) and 
\begin{equation}
\begin{split}
    E_\Phi =& - \nabla_a\nabla^a\Phi \left( 1+\ell^2 \alpha X \right) - \ell^2\nabla^a\Phi\nabla_a(\alpha X) + V' \\
    & - \frac{1}{2}\ell^2\alpha' X^2 - \frac{1}{16} \ell^2 \beta' \delta^{a_1a_2a_3a_4}_{b_1b_2b_3b_4} R_{a_1a_2}{}^{b_1b_2} R_{a_3a_4}{}^{b_3b_4} \, .
\end{split}
\end{equation} 

The coefficient $\beta$ only appears 
through its derivatives in the EoM's, consistent with the fact that the Gauss-Bonnet density by itself is a topological term. 
From a phenomenological perspective, the most natural choices might be $\beta = \beta_0 \Phi$ as it is minimalistic and leads to a shift symmetric theory ($\Phi \to \Phi+c$) for $V\equiv 0$, or $\beta = e^{b_0 \Phi}$. 

The principal symbol captures the structure of the highest derivative terms of the EoM of the EFT. It is made up of the following tensors \cite{Reall:2021voz}
\begin{align*}
    P_{gg}^{abcdef} = \frac{\partial E^{ab}_g}{\partial(\partial_e\partial_f g_{de})}\, ,\; P_{g\Phi}^{abcd} = \frac{\partial E^{ab}_g}{\partial(\partial_c \partial_d \Phi)} \, , 
\end{align*}
\begin{align}    
    P_{\Phi\Phi}^{ab} =\frac{\partial E_\Phi}{\partial(\partial_a\partial_b \Phi)} \, , \; P_{\Phi g}^{abcd}= \frac{\partial E_\Phi}{\partial(\partial_c\partial_d g_{ab})} \, .
\end{align}
In these equations, $\partial_a$ denotes an arbitrary background derivative operator such as, e.g., the flat coordinate derivative associated with some coordinate system.

As shown by \cite{Reall:2021voz}, the above expressions are expressible in terms of certain tensors $C^{abcdef}$, $C^{abcd}$, $P^{ab}$ and the so called \emph{effective metric} $C_{ab}$ via
\begin{subequations}
\begin{align} 
P_{gg}^{abcdef} &= C^{a (c|e b|d)f} \, , \\
C^{abcd} &= P_{\Phi g}^{abcd} \, , \\
P^{ab} &= P_{\Phi\Phi}^{ab} \, , \\
C^{a_1a_2a_3b_1b_2b_3} &= - \frac{1}{2} \epsilon^{a_1a_2a_3 a} \epsilon^{b_1b_2b_3b} C_{ab} \, ,
\end{align}
\end{subequations}
where $\epsilon_{abcd}$ is the volume form of $(\cM,g_{ab})$. 

Ref. \cite[Sec. 4.1]{Reall:2021voz} has shown that the effective metric and the tensors $C^{abcd}$ and $P^{ab}$ of our EFT \eqref{lagrangian1} take the following form:
\begin{subequations}
    \begin{align}
        C_{ab} &= g_{ab} - \ell^2 \nabla_a\nabla_b\beta \\
        C^{abcd} &= - \frac{1}{4} \ell^2 \beta' \epsilon^{abe_1e_2} \epsilon^{cd f_1f_2} R_{e_1e_2f_1f_2}\\
        P^{ab} &= -(1+\ell^2\alpha X) g^{ab} + \ell^2 \alpha \nabla^a\Phi \nabla^{b} \Phi.
    \end{align}
\end{subequations}

\subsection{Propagation Cones of EFT}
\label{sec:propcones}
It has been shown by \cite{harvey2} that the EoMs of the EFT admit a well-posed initial value formulation in a certain gauge, assuming that the theory is suitably ``weakly coupled''. By this one means that, given a 
co-dimension one surface $\Sigma \subset \cM$ that is ``non-characteristic'' in a suitable sense described more fully below, and given suitable initial data 
comprised of $(\Phi, g_{ab})$ and its first derivatives off the surface 
$\Sigma$ subject to the constraints of the theory \eqref{lagrangian}, then one locally (in some open neighborhood of $\Sigma$) has one and only one solution to the EoMs. 

In the case of a quasi-linear system of partial differential equations (PDEs) for some set of fields $\Psi^I$, i.e., the highest derivative 
terms in the EoMs can be written as $P^{Iab}_J(\Psi, \partial \Psi) \partial_a \partial_b \Psi^J + \dots = 0$, one calls a covector $\xi_a$ at a point $p$ \emph{non-characteristic} if the ``principal symbol'', $P^{Iab}_J(\Psi, \partial \Psi) \xi_a \xi_b|_p$,  is an invertible matrix. This means that for a coordinate $x^0$ locally defining $\Sigma$ 
by $x^0=$ const. such that 
$\xi_a|_p = \nabla_a x^0|_p$, we can solve the equation for $\partial_0^2 \Psi^J$ and iteratively determine all derivatives of $\Psi^J$ at $p$ off of $\Sigma$ in terms of 
up to one derivative w.r.t. $x^0$ and the derivatives tangent to $\Sigma$, i.e., w.r.t. some coordinates $x^j, j=1,2,3$ parameterizing $\Sigma$. Thus, if the initial data on $\Sigma$ are smooth, then so will the solution (if it exists) near $p$. 

On the other hand, if a covector $\xi_a$ at a point $p$ is characteristic, which means that $P^{Iab}_J(\Psi, \partial \Psi) \xi_a \xi_b|_p$ is not invertible, then we cannot solve the equation for $\partial_0^2 \Psi^J$. As a consequence, the solution may develop 
discontinuities at characteristic surfaces, i.e., ones such that their (locally) defining function $x^0$ is such that $\nabla_a x^0 = \xi_a$
is characteristic at each point of the surface. In particular, it follows from such a consideration that the edge of the causal shadow of a point, or the boundary of the domain of dependence must be characteristic at each point in the case of a quasi-linear PDE, see e.g., \cite{couranthilbert} for details.
Even if the fields $\Psi^I$ contain a Lorentz metric $g_{ab}$, 
the characteristic surfaces need not have anything to do with the lightcones of this metric; see fig. \ref{dod}.
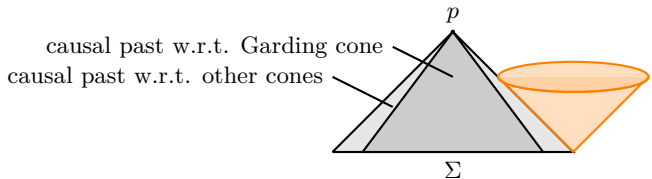
\begin{figure}
\centering
\begin{tikzpicture}[scale=0.2]
    \draw[thick, black, fill=black!10] (0,0) -- (16,0) node[midway, below]{$\Sigma$} -- (8,8) -- (0,0) ;
    \draw[thick, black, fill=black!20] (2,0) -- (14,0) -- (8,8) -- (2,0) ;
    \fill (8,8) circle(4pt) node[above]{$p$} ;
    \draw[thick, black] (8,5) -- (4,7) node[left]{\text{causal past w.r.t. Garding cone}} ;
    \draw[thick, black] (4,3) -- (0,5) node[left]{\text{causal past w.r.t. other cones}} ;

    \draw[thick, orange, fill=orange!50, fill opacity=0.5] (11,5) -- (16,0) -- (21,5) arc[start angle=0, end angle=-180, x radius=5, y radius=1] ;
    \draw[thick, orange, fill=orange!50, fill opacity=0.5] (21,5) arc[start angle=0, end angle=180, x radius=5, y radius=1] ;
    \draw[thick, orange, fill=orange!50, fill opacity=0.7] (21,5) arc[start angle=0, end angle=-180, x radius=5, y radius=1] ;
\end{tikzpicture}
\caption{\justifying The figure shows the causal pasts of a point $p$
w.r.t. various cones terminating on 
some Cauchy surface $\Sigma$. Obviously, the domains of the dependence
will be different.}
\label{dod}
\end{figure}

The EoMs of our EFT \eqref{lagrangian1}, viewed as a system for 
$\Psi^I = (\Phi, g_{ab})$, are not quasi-linear, but one may still 
define a principal symbol by linearizing the EoMs, and one may still define the notion of a characteristic covector. In fact, 
as in \cite[Sec. 3.2]{Reall:2021voz} we say that a covector $\xi_a$ is \emph{characteristic} if there exists a \emph{non-trivial} (i.e., non-gauge) $T=([t_{ab}],t)$ such that $P(\xi)T=0$, where
    \begin{align}
    \label{Pdef}
        P(\xi)T =
        \begin{pmatrix}
            P_{gg}^{abcdef} \xi_e\xi_f & P_{\Phi g}^{abef}\xi_e\xi_f \\
            P_{g\Phi}^{cdef} \xi_e\xi_f & P_{\Phi\Phi}^{ef}\xi_e\xi_f
        \end{pmatrix}
        \begin{pmatrix}
            t_{cd} \\
            t
        \end{pmatrix}
    \end{align}
    and $P(\xi)$ is called the \emph{principal symbol} of $E_g^{ab}=E_\Phi=0$. The equivalence class is defined by $t_{ab} \sim t'_{ab}$ if there exists a covector $X_a$ such that $t'_{ab}=t_{ab} + \xi_{(a}X_{b)}$, noting that $\xi_{(a}X_{b)}$
    always is in the kernel of $P$. Such elements correspond to infinitesimal gauge transformations. 
    A hypersurface is $\Sigma \subset \cM$ is said to be a \emph{characteristic hypersurface} if its co-normal is 
    characteristic.

Just as for quasi-linear equations, the edge of the domain of dependence of a suitably gauge-fixed \cite{harvey2} version of the linearized EoMs is characteristic in the above sense. It seems plausible to us that this will remain true even in the full non-linear EFT \eqref{lagrangian1}, at least in the \emph{weakly coupled} case. 
(We call the solution of our EFT weakly coupled if the contributions of the 4-derivative terms to $P(\xi)$ are small compared to the 2-derivative contributions, see \cite[Sec. 4.1]{Reall:2021voz}.) We leave this as an open problem and use the above notion of characteristic covector to define the propagation cone of the theory, which, at any rate, will be reasonable for the solutions of the linearized theory, and which should be a reasonable approximation for small $\ell$.

Ref. \cite[Sec. 4.3]{Reall:2021voz} showed that there is an equivalent condition for $\xi_a$ to be characteristic in the following way:
Let $\xi_a$ be a non-zero real covector. Then $\xi_a$ is characteristic if and only if \cite[Sec. 3.2]{Reall:2021voz}
\be 
p(\xi)=C^{-1}(\xi) Q(\xi)=0, 
\ee
where $C^{-1}(\xi)=(C^{-1})^{ab} \xi_a\xi_b$ and
\begin{subequations}
    \begin{align}
    \label{Qdef}
        Q(\xi) &= \frac{1}{4!}Q^{abcd} \xi_a\xi_b\xi_c\xi_d \, ,  \\
        Q^{abcd} &= \det (Cg^{-1}) (C^{-1})^{(ab}P^{cd)} + \\ 
        &(2C_{a_1a_2}C_{b_1b_2}- C_{a_1b_1}C_{a_2b_2})C^{a_1(ab|a_2} C^{b_1|cd)b_2} \, . \nonumber
    \end{align}
\end{subequations}
Note that in the case of a 2-derivative theory [standard Einstein-scalar having $\ell=0$ in \eqref{lagrangian1}] we have $(C^{-1})^{ab}=-P^{ab}=g^{ab}$ and $C^{abcd}=0$ \cite[Sec. 4.4]{Reall:2021voz}, so that the \emph{characteristic cone} at a point $x\in\cM$, $\{\xi_a \in T_x^*\cM : p(\xi)|_x=0\}$, is simply the null cone of the physical metric $g^{ab}|_x$ in $T^*_x\cM$. Therefore, a covector $\xi_a$ in a 2-derivative theory is characteristic if and only if $\xi_a$ is null w.r.t. $g^{ab}$
In the following, we refer to $\{\xi_a\in T_x^*\cM : C^{-1}(\xi)|_x = 0\}$ as the \emph{quadratic cone} and to $\{\xi_a\in T_x^*\cM : Q(\xi)|_x =0\}$ as the \emph{quartic cone}. As  illustrated in ref. \ref{fig:ccs}, both cones are subsets of the characteristic cone and the quartic cone consists of two sheets between which the quadratic cone lies \cite[Sec. 4.4]{Reall:2021voz}. 

\begin{figure}
\centering
\begin{tikzpicture}[scale=0.25]   
    \draw[very thick, red, fill=red!50, fill opacity=0.5] (7,7) arc[start angle=0, end angle=180, x radius=7, y radius=1.5] ;
    \fill[red!60, fill opacity=0.5] (7,7) arc[start angle=0, end angle=-180, x radius=7, y radius=1.5] ;
    \draw[very thick, blue, fill=blue!50, fill opacity=0.5] (7,7) arc[start angle=0, end angle=180, x radius=5.5, y radius=1.1] ;
    \fill[blue!60, fill opacity=0.5] (7,7) arc[start angle=0, end angle=-180, x radius=5.5, y radius=1.1] ;
    \draw[very thick, orange, fill=orange!50, fill opacity=0.5] (4,7) arc[start angle=0, end angle=180, x radius=4, y radius=0.8] ;
    \fill[orange!60, fill opacity=0.5] (4,7) arc[start angle=0, end angle=-180, x radius=4, y radius=0.8] ;
    \draw[very thick, orange, fill=orange!40, fill opacity=0.7] (-4,7) -- (0,0) -- (4,7) arc[start angle=0, end angle=-180, x radius=4, y radius=0.8] ; 
    \draw[very thick, blue, fill=blue!40, fill opacity=0.7] (-4,7) -- (0,0) -- (7,7) arc[start angle=0, end angle=-180, x radius=5.5, y radius=1.1] ;
    \draw[very thick, red, fill=red!40, fill opacity=0.7] (-7,7) -- (0,0) -- (7,7) arc[start angle=0, end angle=-180, x radius=7, y radius=1.5] ;

    \fill[fill=black!15, fill opacity=0.35] (-13,4) -- (5,4) -- (13,10) -- (-5,10) -- (-13,4)  ;
    \draw[very thick, black] (-13,4) -- (5,4) -- (13,10) node[above, black!100] {$u^a\xi_a=$const.} -- (-5,10) -- (-13,4)  ;

    \draw[very thick, dashed] (0,0) -- (0,7) ;
    \draw[very thick, -Latex] (0,7) -- (0,13) node[above] {$u^a$} ;

    \end{tikzpicture}
\caption{\justifying Generic behavior of the characteristic cone for a non-vanishing Weyl-like tensor \eqref{Wdef} \cite{Reall:2021voz}. The red and and the orange cone illustrate the two sheets of the quartic cone. The blue cone is the quadratic part of the characteristic cone, which smoothly touches the quartic cones along two principal null directions of the Weyl-like tensor.}
\label{fig:ccs}
\end{figure}
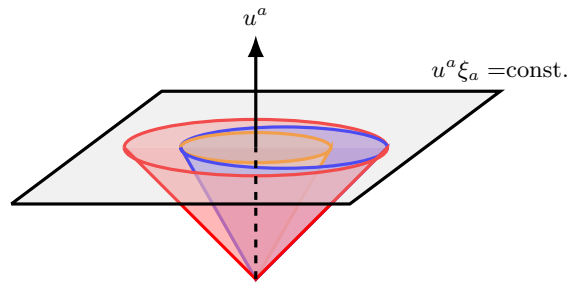
The characteristic cone provides the following notion of causality in the cotangent space:
\begin{definition}
    \cite[Sec. 4.4]{Reall:2021voz}: Let $x\in\cM$. Assume that the theory is weakly coupled (such that $g^{ab}|_x \approx (C^{-1})^{ab}|_x$) and fix a future pointing\footnote{We assume a time orientable setting in which a continuous choice of $u_a$ can be made globally.} causal covector $u_a$ in $T_x^*\cM$ w.r.t. $(C^{-1})^{ab}|_x$. We call the \emph{connected component} of $\{\xi_a \in T_x^*\cM : p(\xi)|_x \neq 0\}$ that contains $\mp u_a$ the \emph{Garding cone} $\Gamma^\pm_x$ at the point $x\in\cM$.
\end{definition}
The Garding double cone $\Gamma^+_x \cup \Gamma^-_x$ is the region inside the innermost sheet of the quartic cone, i.e., in a 2-derivative theory, it is simply the set of covectors that are timelike w.r.t. the physical metric at a point \cite[Sec. 4.4]{Reall:2021voz}. Starting from the Garding double cone, one can define a notion of causality in the tangent space:
\begin{definition}
    \cite[Sec. 4.6, 4.7]{Reall:2021voz}: For $x\in\cM$ define the \emph{causal cone} $\cC^\pm_x$ in $T_x\cM$ as the dual of $\Gamma^\pm_x$, i.e.,
    \begin{align}
        \cC^\pm_x = \{ X^a \in T_x\cM : \xi_a X^a |_x \leq 0 \, \forall \xi_a \in \Gamma^\pm_x \} \, .
    \end{align}
    A vector $X^a\in T_x\cM$ is said to be
    \begin{itemize}
        \item future- (past-) directed causal at $x$ if and only if $X^a \neq 0$ and $X^a \in \cC^+_x$ ($\cC^-_x$),
        \item future- (past-) directed timelike at $x$ if and only if $X^a \in \overset{\circ}{\cC_x^+}$ ($\overset{\circ}{\cC^-_x}$) \, .
    \end{itemize}
    Here, a circle above a set means its interior.
    
    A vector field $X^a$ is called future- (past-) directed causal (timelike) if $X^a|_x$ is future- (past-) directed causal (timelike) for all $x\in\cM$. A smooth curve $\gamma$ on $\cM$ is said to be future- (past-) directed causal (timelike) if its tangent vector field is future- (past-) directed causal (timelike). Furthermore, let $S$ be some subset of $\cM$. The sets
    \begin{widetext}
    \begin{subequations}
    \begin{align}
    \label{IJdef}
        J^\pm(S) &= \{ x \in \cM : \text{there exists a future- (past-) directed causal curve from $S$ to $x$}\} \, , \\
       I^\pm(S) &= \{x \in \cM : \text{there exists a future- (past-) directed timelike curve from $S$ to $x$}\} \, .
    \end{align}
    \end{subequations}
    \end{widetext}
    are called \emph{causal future (past)} and \emph{chronological future (past)}.
\end{definition}

\subsection{Slowness Surface of EFT}
\label{sec:slowness}
Taking any cross-section of constant $\xi_a u^a$ of the characteristic cone at a point, one obtains the so-called \emph{slowness surfaces} \cite{Reall:2021voz}, consisting of three ellipsoids where the innermost one corresponds to the "fastest degree of freedom" (note the similarity with the null cone in 2-derivative theories) \cite[Sec. 4.4]{Reall:2021voz}. On this surface there are special points at which two or more sheets of the slowness surface touch each other as shown in fig. \ref{fig:ccs}. Such points on the slowness surface then correspond to a special direction of the characteristic cone by scaling freedom. They may be classified in the following way:
\begin{definition}
\label{def:II.3}
    \cite[Sec. 4.4]{Reall:2021voz}: Suppose that $\xi_{0a}$ belongs to the characteristic cone. Then $\xi_{0a}$ is called
    \begin{itemize}
        \item a \emph{singular direction} if $\partial p/\partial \xi_a |_{\xi=\xi_0}=0$ or, equivalently, $\xi_{0a}$ belongs to the quadratic and the quartic cone, 
        \item a \emph{double direction} if  $\xi_{0a}$ is a singular direction of the characteristic cone and $\partial^2p/\partial\xi_a\partial \xi_b|_{\xi=\xi_0}\neq0$,
        \item a \emph{triple direction} if $\xi_{0a}$ is a singular direction of the quartic cone and $\partial^2p/\partial\xi_a\partial \xi_b|_{\xi=\xi_0}=0$.
    \end{itemize}
\end{definition}
Generically, at a triple point all three sheets of the slowness surface coincide, so the propagation speed, interpreted as the distance from the origin, is equal for all three polarizations \cite[Sec. 4.4]{Reall:2021voz}.
At a double point, only two sheets of the slowness surface---one of the quadratic- and one of a quartic cone---touch \cite[Sec. 4.4]{Reall:2021voz}, see fig. \ref{fig:ccs}.

Within the expression for $Q(\xi)$, one can split off a part that is determined by a "Weyl-like" tensor $W^{abcd}$ given by 
\begin{align}
    W^{abcd} =& C^{abcd} - (C^{-1})^{a [c} D^{d]b} + \nonumber
    \\&(C^{-1})^{b[c}D^{d]a} + \frac{1}{3} D (C^{-1})^{a[c} (C^{-1})^{d]b} \, ,
    \label{Wdef}
\end{align}
where $D^{ab}$ is the "Ricci-like tensor" and $D$ is the "Ricci-like scalar" formed from $C^{abcd}$ using the effective metric $C_{ab}$ \cite[Sec. 3.2]{Reall:2021voz}. Ref. \cite[Sec. 4.4]{Reall:2021voz} showed that a singular direction $\xi_{0a}$ must be a principal null direction of this Weyl-like tensor.

\subsection{Bicharacteristic Curves of EFT}
We have seen that the innermost sheet $\partial(\Gamma^+ \cup\Gamma^-)$ of the quartic cone plays the role of the fastest propagation direction. Thus, in analogy to the null cone in 2-derivative theories, we define the event horizon $\sH^\pm=\partial(\cM \backslash J^\mp (\mathscr I^\pm ))$ just as in \eqref{futurehor}, but where the notions $J^\pm$ of the causal future (past) are now defined as above instead of by the physical metric. Thus, a horizon is a characteristic hypersurface with normal $\xi_a\in \partial\Gamma^\pm$. 

A notion appearing in the assumptions of Theorem A below which we have not reviewed yet is that of a bicharacteristic curve.
\begin{definition}
\label{bichardef}
    \cite[Sec. 4.5]{Reall:2021voz}: A \emph{bicharacteristic curve of $Q$} \eqref{Qdef} is a pair $(x(s),\xi_a(s))$ satisfying Hamilton's equations 
    \begin{align}
    \label{HF}
        \dot{x}^a = \frac{\partial Q}{\partial\xi_a} \, , \; \; \dot \xi_a=- \frac{\partial Q}{\partial x^a} \, 
    \end{align}
    with the initial condition that $(x(0),\xi_a(0))$ belongs to the quartic cone. Here a dot means derivative w.r.t. $s$, and $x(s)$ is called the \emph{projection of a bicharacteristic curve}.
\end{definition}

If we start a bundle of bicharacteristic curves on a co-dimension one 
surface on a non-characteristic 3-surface $\Sigma$, then these curves 
locally generate a characteristic surface. In particular, the future/past horizons $\mathscr{H}^\pm$ locally are ruled by bicharacteristic curves. 

Bicharacteristic curves of $Q$, see \eqref{Qdef}, in 4-derivative theories are in close analogy to null geodesics of $g_{ab}$ in 2-derivative theories which would be obtained from the Hamiltonian $H=g^{ab}\xi_a \xi_b$ instead of $Q$. For example, the future horizon (usually generated by null geodesics in standard GR) is generated by bicharacteristic curves of the innermost quartic sheet $\partial\Gamma^\pm$ in 4-derivative theories. 

A difference between the usual case of null geodesics is that there is no canonical parameterization of the bicharacteristics associated with $Q$ analogous to affine parameterization in general. Furthermore, it may---and actually, in our situation below, will---happen that the Hamiltonian flow is trivial in the sense that $\dot x^a=0$ on the characterstic surface.

\subsection{Killing Horizons in EFT}
\label{sec:KH}
A \emph{Killing horizon} in a spacetime $(\cM, g_{ab})$ with a Killing vector field $\chi^a$ is a co-dimension one null surface $\cN$ such that 
$\chi^a$ is both tangent and normal (meaning $\chi_a v^a = 0$ for all 
$v^a \in T\cN$) to $\cN$. Such a Killing vector field $\chi^a$ automatically satisfies 
\begin{equation}
\label{Keq}
    \chi^a \nabla_a \chi^b = \kappa \chi^b \quad 
    \text{on $\cN$,}
\end{equation}
where $\kappa$ is called the \emph{surface gravity}. Eq. \eqref{Keq} immediately gives $\chi^a \nabla_a \kappa = 0$ on $\cN$. Furthermore, it is known \cite{Waldbook} that if the Einstein tensor $G_{ab}$ of $g_{ab}$ satisfies the dominant energy condition, then $\kappa$ is actually constant on all of $\cN$. This property is called the zeroth of law black hole thermodynamics in the case of black holes with Killing horizons. 

In a theory with scalars, it is natural to expect that the zeroth law includes a statement that $\Phi$ is constant on the horizon. As \cite[p.31]{Reall:2021voz} showed, the zeroth law holds in this sense within the setting of our EFT \eqref{lagrangian1} for any weakly coupled solution $(\Phi, g_{ab})$ to the EoMs that are Lie-derived by $\chi^a$, 
\begin{equation}
\label{gPLieder}
    \sL_\chi g_{ab} = \sL_\chi \Phi = 0
\end{equation}
for any $\cN$ which is a Killing horizon relative to $\chi^a$.
As a corollary\footnote{Ref. \cite{Hollands:2022ajj} shows more generally that the horizon actually is a Killing horizon w.r.t. some Killing vector field $\chi^a$ that is not necessarily assumed from the beginning.} of the rigidity theorem for general EFTs by \cite{Hollands:2022ajj}, this property holds, in fact, much more generally for the horizon $\mathscr{H}$ defined w.r.t. $g_{ab}$ of any stationary, asymptotically flat, real analytic, non-degenerate $(\kappa > 0)$ black hole with compact horizon cross sections, of \emph{any} EFT constructed from a diffeomorphism covariant Lagrangian, to arbitrary orders in the EFT scale $\ell$. 
See also \cite{Racz:1995nh} for model-independent proofs of the zeroth law based on a $t$-$\varphi$ reflection symmetry.

If we have two Killing horizons ${}^{L,R}\cN$ with the same $\chi^a$ and surface gravity $\kappa>0$, intersecting in a (compact) 2-surface $\mathscr{B}$, one speaks of a non-degenerate (compactly generated) \emph{bifurcate Killing horizon} \cite{Kay:1988mu}, see fig. \ref{fig:biKi}. (A)deSitter-, Minkowski, the extended subextremal Kerr- and Schwarzschild spacetimes including their (A)dS counterparts all host such structures.
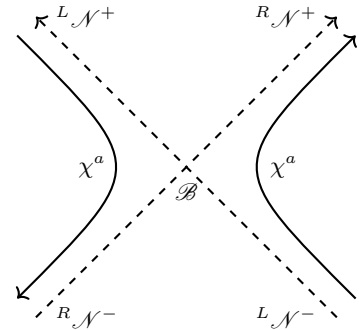
\begin{figure}[h!]
\centering
\begin{tikzpicture}[scale=.5]
\draw[->, thick, dashed] (-2,-3) -- (6,5);
\draw[->,thick, dashed] (6,-3) -- (-2,5);
\draw (-1.7,5) node[anchor=west]{${}^L\cN^+$};
\draw (5.7,5) node[anchor=east]{${}^R\cN^+$};
\draw (-1.7,-3) node[anchor=west]{${}^R\cN^-$};
\draw (5.7,-3) node[anchor=east]{${}^L\cN^-$};
\draw (4,1) node[anchor=west]{$\chi^a$};
\draw (-1.1,1) node[anchor=west]{$\chi^a$};
\draw (2,.8) node[anchor=north]{${\mathscr B}$};
\draw[<-,thick] (-2.5,-2.5) .. controls (1,1)  .. (-2.5,4.5);
\draw[->,thick] (6.5,-2.5) .. controls (3,1)  .. (6.5,4.5);
\end{tikzpicture}
\caption{\justifying Bifurcate Killing horizon.}
\label{fig:biKi}
\end{figure}

\begin{proposition}
\label{prop1}
    Let $\cN$ be a non-degenerate $(\kappa > 0)$ bifurcate Killing horizon with Killing vector field $\chi^a$
    which Lie-derives the solution $(\Phi,g_{ab})$ of the EFT \eqref{lagrangian1}. Then $\cN$ is characteristic 
    w.r.t. the quartic and quadratic cone, 
    $Q^{abcd}\xi_a \xi_b \xi_c \xi_d = 0 = (C^{-1})^{ab} \xi_a \xi_b$
    for all co-normals $\xi_a$ to $\cN$. Furthermore, 
    \begin{equation}
    \label{Qdef1}
        \frac{\partial Q}{\partial \xi_a}\bigg|_{\cN} = 0,
        \quad
    \end{equation}
    for all characteristic covectors $\xi_a$
    at $\cN$, all of which are singular triple directions in the sense of def. \ref{def:II.3}.
\end{proposition}

\emph{Proof:} The proof is quite similar in spirit to arguments by \cite{Benkel:2018qmh, Tanahashi:2017kgn,Izumi:2014loa}. 
By the definition of Killing horizon, $\chi_a$ is co-normal to 
$\cN$. Thus, we must show $Q^{abcd}\chi_a \chi_b \chi_c \chi_d = 0 = (C^{-1})^{ab} \chi_a \chi_b$ on $\cN$. Since $Q_{abcd}, C_{ab}$ are locally and covariantly constructed out of $(\Phi,g_{ab})$, we 
have 
\begin{equation}
    \sL_\chi Q^{abcd} = \sL_\chi (C^{-1})^{ab} = 0.
    \label{LieQC}
\end{equation}
Since $\sL_\chi \chi_a=0$, we thereby have $\sL_\chi[Q^{abcd}\chi_a \chi_b \chi_c \chi_d]=0, 
\sL_\chi[(C^{-1})^{ab} \chi_a \chi_b]=0$. Thus, the quantities in square brackets are constant along the orbits of $\chi^a$. However, $Q^{abcd}\chi_a \chi_b \chi_c \chi_d = 0 = (C^{-1})^{ab} \chi_a \chi_b$ on $\mathscr{B}$, since $\chi^a$ vanishes there and since $Q^{abcd}, (C^{-1})^{ ab}$ are smooth, in particular finite, at $\mathscr{B}$. 

To show \eqref{Qdef1}, we restrict attention to ${}^R \cN^+$, say, and introduce a second null vector field $N^a$ uniquely defined by \cite[sec. II.B]{Hollands:2024vbe}
\begin{equation}
    \nabla_a \chi_b = 2\kappa N_{[a} \chi_{b]}.
\end{equation}
at points of ${}^R \cN^+$. By uniqueness, we have $\sL_\chi N^a = 0$. We complete $\chi^a, N^a$ to a null tetrad by introducing two additional vectors $s^a_i, i=1,2$ such that $\sL_\chi s^a_i = 0, s_i^a s_{ja} = \delta_{ij}, s^a_i N_a = s^a_i \chi_a = 0$. By \cite[sec. II.B]{Hollands:2024vbe}, the limits of $s_{ia}, N_a \chi_b$ towards $\sB$ exist, 
and $Q^{abcd}N_a \chi_b \chi_c \chi_d$, as well as 
$Q^{abcd}s_{ia} \chi_b \chi_c \chi_d$, are Lie-derived by $\chi^a$. Since the limit of $\chi^a$
at $\sB$ vanishes, these quantities vanish on ${}^R \cN^+$, and by a similar argument for the other parts of $\cN$, on all of $\cN$. Since we already know that $Q^{abcd}\chi_a \chi_b \chi_c \chi_d = 0$ on $\cN$, the proof of $\eqref{Qdef1}$ is complete. 

In view of $p(\xi)=C^{-1}(\xi)Q(\xi)$, the previous results imply that $\chi_a$ is a triple singular direction. 
\qed


\medskip
Note that the proof used neither the detailed structures of $Q_{abcd}, C_{ab}$, nor the fact that $(\Phi, g_{ab})$ is a solution to the EoMs. In this paper, we are interested in the \emph{converse} of prop. \ref{prop1}, i.e., we want to show 
that a horizon defined w.r.t. the propagation cones of our EFT
is actually a Killing horizon. This will require the detailed structure of the propagation cone in our EFT and the EoM's in secs. \ref{sec:thmA}, \ref{sec:thmB} and app. \ref{app:proofA}.

\subsection{Generalized Gaussian Null Coordinates}
\label{sec:gGNC}
Whether or not $g_{ab}$ satisfies any EoMs or other special equations, one can set up a convenient coordinate system in the neighborhood of any (part of a) Killing horizon referred to as \emph{Gaussian null coordinates (GNCs)}. We now review this construction; see, e.g., \cite{Hollands:2006rj} for details. We assume that the Killing vector field $\chi^a$ 
is complete on $\cN$, which we should think of as 
being either ${}^{R,L}\cN^+$ in this subsection. Thus, its flow can be used to set up a diffeomorphism between $\cN$
and $N \times \mathbb{R}$ such that the flow of $\chi^a$ is along the $\mathbb{R}$ Cartesian factor. We refer to $N$
as a \emph{cut} of $\cN$; the orbits of $\chi^a$ must be transversal to this $N$.

First, we define local coordinates $x^A,A=1,2$ covering some subset of $N$. Let $\{F_v\}$ be the 1-parameter flow of isometries associated to $\chi^a$ and define $N(v):=F_v[N]$, i.e., on $\cN$ we have $\chi^a=(\partial_v)^a$. Finally, at each point of $\cN$, define a null vector field $l^a$ satisfying $l^a\chi_a=1$ and $l^am_a=0$ for all $m^a\in TN(v)$. By construction, $l^a=(\partial_u)^a$. Let $u$ be the affine parameter of null geodesics off of each $N(v)$ with initial velocity $l^a$.

Then in $(v,u,x^A)$-coordinates, the metric takes the form \cite{Hollands:2006rj} 
\begin{align}
\label{gGNC}
    g = -f\ud v^2 + 2\ud v\ud u + 2 k_A \ud x^A \ud v + h_{AB} \ud x^A \ud x^B \, ,
\end{align}
and $u$ is a defining coordinate for $\cN$. 

We define the following tensors:
\begin{align}
\label{hdef}
    h_{ab} = h_{AB} (\ud x^A)_a (\ud x^B)_b , \; k_a=k_A (\ud x^A)_a ,
\end{align}
and $q^{ab},p^a{}_b$ are defined by
\begin{align}
\label{pqdef}
    p^a{}_b= (\partial_A)^a(\ud x^A)_b , \; q^{ab} = (h^{-1})^{AB} (\partial_A)^a(\partial_B)^b,
\end{align}
These tensors are independent of the chosen local coordinates $x^A$ on $N$
and thereby globally defined near $\cN$. They satisfy
\begin{subequations}
\begin{align}
    &p^a{}_bk_a=k_b, \; p^a{}_b p^c{}_d h_{ac}=h_{bd}, \; p^a{}_b p^b{}_c = p^a{}_c,\\
    &p^a{}_b = q^{ac} h_{cb}, \; p^a{}_bl^b =p^a{}_b \chi^b=0,\\
    &\sL_\chi k_a = \sL_\chi f = \sL_\chi l^a = \sL_\chi q^{ab} = 
\sL_\chi p^a{}_b = 0,
\end{align}
\end{subequations}
near $\cN$, as well as 
\begin{equation}
k_a|_\cN = f|_\cN = 0, \; p^a{}_b|_\cN = g_{bc}q^{ca}|_\cN.
\end{equation}
The surface gravity $\kappa$ is given by
\be
\kappa = \frac{1}{2} \sL_l f|_{\cN}.
\ee
Conversely, if a metric is defined by a set of coordinates  
$(v,u,x^A)$ as in \eqref{gGNC} with metric components that are independent of $v$, then $u=0$ defines a Killing horizon $\cN$
with the Killing vector field $\chi^a = (\partial_v)^a$, provided $f,k_a=0$ on $\cN$.

The above construction can be repeated in an analogous manner also in the case that $\chi^a$ is tangent to $\cN$ but not necessarily normal, as would have to be the case if $\cN$ was not a null surface. This construction, which we will need to apply to a characteristic surface $\cN$ of the Garding cone below---not a priori known to be null---again yields coordinates $(u,v,x^A)$ such that $g_{ab}$ satisfies \eqref{gGNC}, and tensors $f,k_a,h_{ab},p^a{}_b,q^{ab}$. If the pull-back of $g_{ab}$ to $N$ is a Riemannian metric, as will always be the case for sufficiently small $|\ell|$ below, then so is $h_{ab}$, as for a Killing horizon. However, contrary to the case of a Killing horizon, $f, k_a$ will \emph{not} in general vanish on $\cN$. We will refer to these $(u,v,x^A)$ as \emph{generalized Gaussian null coordinates} (gGNC) in the following.

Below, we will consider real analytic 1-parameter families of 
spacetimes $(\cM(\ell), g_{ab}(\ell,x))$, surfaces $\cN(\ell)$ tangent to a Killing field $\chi^a(\ell)$, cuts $N(\ell)$, as above etc. Then for each $\ell$, we can introduce gGNC based on $N(\ell)$, which cover an open neighborhood of each $\cN(\ell)$. We may then introduce an identification (diffeomorphism) between these neighborhoods by declaring points for different $\ell$ to be equal if their gGNCs coincide. After such an identification ("gauge choice"), we are dealing with only \emph{one} spacetime, in which the locations of $\cN$ and $N$, or $\chi^a$, no longer vary with $\ell$.

It should be noted, however, that if $\cN$ was e.g. a bifurcate Killing horizon (fig. \ref{fig:biKi}) for the $\ell=0$ member of the family of metrics, then the metrics $g_{ab}(\ell, x)$ will \emph{not} be smooth at $\sB$ if $\cN$ is not a Killing horizon (e.g., if it is not null) for $\ell \neq 0$.

Below, it is convenient to use a ``purely angular'' derivative operator $D_a$ associated with gGNCs acting on purely angular tensors $T_{ab \dots c}$, i.e., ones which are projected by $p^a{}_b$. Its definition is
\begin{equation}
\label{Ddef}
    D_b T_{a_1 \dots a_r} = p^c{}{}_b p^{d_1}{}_{a_1} 
    \cdots p^{d_r}{}_{a_r} \nabla_c T_{d_1 \dots d_r}.
\end{equation}
Then we have
\begin{equation}
    0=[\sL_\chi, D_a] =D_a h_{bc} = D_a q^{bc} = D_a p^b{}_c = 0.
\end{equation}
The components of the Riemann tensor $R_{abcd}$, contracted into 
$p^a{}_b, \chi^a,$ or $l^a$ in arbitrary ways, may thereby be expressed 
in terms of contractions of $k_a, f$, their $D_a$ covariant derivatives 
and $p^a{}_b, q^{ab}, \chi^a, l^a$. These expressions are provided in app. \ref{app:GNC} for later use.

\section{Theorem A}
\label{sec:thmA}

\subsection{Assumptions}
For Theorem A, we assume the following: 
    \begin{enumerate}
        \item[(i)] 
        $(\cM(\ell), \Phi(\ell,x),g_{ab}(\ell,x))$ is a 1-parameter family of solutions to the EFT \eqref{lagrangian1} that is jointly real analytic in $(x,\ell)$ under some 1-parameter family of diffeomorphisms identifying each $\cM(\ell)$ with some fixed $\cM$.
        \item[(ii)] For $\ell=0$, the solution has an ordinary, non-degenerate (i.e. $\kappa > 0$)
        Killing horizon, $\cH$, with Killing vector field $K^a$. 
        \item[(iii)] There is a family $\cH(\ell)$ for $|\ell|<\ell_0$ ruled by (projections of) bicharacteristic curves of $Q$ 
        i.e., $\cH(\ell)$ is characteristic (co-tangent to the Garding cone of $(\Phi(\ell), g_{ab}(\ell)$).  
        \item[(iv)] $\cH(\ell)$ is analytically diffeomorphic to $H \times {\mathbb R}$ via a family of analytic diffeomorphisms depending analytically on $\ell$. $H$ is diffeomorphic to $\mathbb{S}^2$. 
        \item[(v)] There are commuting vector fields $t^a(\ell,x),\phi^a(\ell,x)$, jointly analytic in $(x,\ell)$, 
        which Lie-derive \eqref{gPLieder} the solution $(\Phi(\ell,x),g_{ab}(\ell,x))$ for all $\ell$. $\phi^a(\ell)$ generates an isometric action of $U(1)$ and is
        tangent to the $H$ Cartesian factor in (iv). $(H,h_{ab})$ has no Killing vector fields other than $\phi^a$.
        $t^a(\ell)$ is complete on $\cH(\ell)$ and may be chosen to be tangent to the $\mathbb R$ Cartesian factor in (iv).
        
        There is a linear combination $\chi^a(\ell) = t^a(\ell) + \Omega(\ell) \phi^a(\ell)$ such that $\dot{x}^a = c \chi^a$ i.e. it is tangent 
        to the bicharacteristics of $Q$ (see def. \ref{bichardef}), and either $c = 0$ or $c \neq 0$
        on all of $\cH(\ell)$.
        \item[(vi)] $\beta'$ has no zeros. 
    \end{enumerate}
See sec. \ref{sec:setup} for notions such as quartic cone, Garding cone, bicharacteristics etc. referring to the propagation of signals in our EFT.
    
Before we state Theorem A, we comment on these assumptions. Analyticity in (i) is just a technical assumption which is needed since our proof works order-by-order in $\ell$.
If we were to assume that the $\ell=0$ solution describes a stationary, analytic, asymptotically Minkowskian, non-degenerate black hole, then the Killing horizon structure would be a consequence of \cite{Hollands:2006rj}. Here, we prefer not to make such a global assumption. 

Assumption (iii) corresponds to the fact that, for an asymptotically Minkowskian or AdS spacetime, the (future) horizon $\cH$ would be the boundary of 
$\cM \setminus J^-(\sI^+)$, where $J^-$ is defined by the Garding cone \eqref{IJdef}. 
In this setting, (iii) would simply be a consequence of this fact. Again, because our constructions related to Theorem A are entirely local, we prefer not to require the existence of $\sI^+$, and to state our assumptions as local conditions on $\cH$.

Assumption (iv) is required to set up our gGNC system \eqref{gGNC1} in the proof and is of a purely technical nature. As already described in sec. \ref{sec:gGNC}, we will set up the gGNC based on $H(\ell)$---the image of $H$ under the identification in (iii)---
for each $\ell$ based on some choice of Killing vector field $K^a(\ell)$ (compatible with $\phi^a(\ell)$) with complete orbits for each $\ell$. Afterwards, we make an identification of the spacetimes $\cM(\ell)$ based on these gGNCs. In this "gauge", the spacetimes, $K^a, \phi^a, \cH,$ etc. will then be \emph{independent} of $\ell$. Analyticity of the identification is required for the same reason as in (i). 

Assumption (v) is the essential and non-trivial one, in the following sense. We know \cite{Hollands:2022ajj} that for any stationary, analytic, asymptotically flat, non-degenerate black hole with horizon defined relative to $g_{ab}$, and in \emph{any} EFT, there exist $K^a, \phi^a, \omega$ such that 
$K^a = t^a + \omega \phi^a$ is tangent to the null generators of this horizon (different a priori from $\cH$), i.e. $\dot x^a \propto K^a$ for the ordinary characteristic curves (null-geodesics) with tangent $\dot x^a$.
We are effectively claiming that a similar statement remains true with regard to the charactersitics defined by the propagation cone (quartic cone) in our EFT, for which $K^a$ need not be equal to $\chi^a$ and $\Omega$ need not be equal to $\omega$ a priori. 

\subsection{Statement}
\label{sec:statement}
Theorem A is the following statement:

\begin{theorem}
    \label{thmA}
Under the assumptions (i)--(vi), $\cH$ is an ordinary Killing horizon with Killing vector field $\chi^a$  w.r.t. the metric $g_{ab}$ for sufficiently small $|\ell|<\ell_0$.
\end{theorem}

\emph{Remarks:} (1) We have already mentioned in sec. \ref{sec:setup} that the surface gravity $\kappa$ of a Killing horizon in our EFT must be constant, as must be $\Phi$. By general results about Killing horizons, $\chi^a$ is furthermore a repeated principal null direction of the Weyl tensor\footnote{By the results of \cite[Sec. 4.4]{Reall:2021voz}, it is at the same time a principal null direction of the Weyl-like tensor \eqref{Wdef}.} of $g_{ab}$, see e.g., \cite[sec. 12.5]{Waldbook}.

(2) In the following all statements are valid for sufficiently small $|\ell|<\ell_0$, where $\ell_0$
is some positive number, and we will no longer repeat this.

\medskip

We can now apply the results of \cite[Sec. 4.7]{Reall:2021voz} or prop. \ref{prop1} to conclude:
\begin{corollary}
\label{cor:1}
    $\chi_a$ is a triple direction (see sec. \ref{sec:setup}) of the characteristic cone on $\cH$, and
    \begin{align}
        \left.\frac{\partial^2 Q}{\partial\xi_a\partial\xi_b}\right|_{\xi=\chi} \propto \chi^a\chi^b \, .
    \end{align}
In particular, the two sheets of the quartic cone and that of the quadratic cone all touch at points of $\cH$ in the co-tangent vector $\chi_a$. See fig. \ref{fig:1} for an illustration.
\end{corollary}

\subsection{Outline of Proof}
\label{sec:thmAproofoutline}
For the proof, we construct gGNCs around $\cH$ just as in sec. \ref{sec:KH}. In the present case, $\cH$ is not known a priori to be null w.r.t. $g_{ab}$, but the construction can nevertheless 
be carried just the same. To show that $\cH$ is a Killing horizon, we must show that $f=k_a=0$ on $\cH$. 

We expand all tensor fields or constants as power series in $\ell$, as e.g., in 
\begin{subequations}
\begin{align}
\label{power_series}
    k_a =& \sum_{j=0}^\infty \ell^j k_{(j)a},\\
    f =& \sum_{j=0}^\infty \ell^j f_{(j)},\\
    \Omega =& \sum_{j=0}^\infty \ell^j \Omega_{(j)},
\end{align}
\end{subequations}
etc. These expansions are  absolutely convergent by our analyticity assumptions. In our proof, we set up our gGNC system as in \eqref{gGNC} at first not w.r.t. $\chi^a = t^a + \Omega\phi^a$, but initially instead for 
\be
K^a = t^a + \Omega_{(0)} \phi^a . 
\ee
In the course of our inductive proof provided in detail in app. \ref{app:proofA}, we will update the definition of $K^a$ at each subsequent order in $\ell$ taking into account corrections to $\Omega_{(0)}$. By assumption, $K^a$ Lie-derives
$(\Phi, g_{ab})$. Furthermore, by assumption, $\cH$ is a Killing horizon for the zeroth order metric $g_{(0)ab}$ w.r.t. $K^a$, and as a consequence, we get that 
\be
f_{(0)} |_{\cH} = k_{(0)a}|_{\cH} = 0.
\ee
The idea of the proof is to verify the analogous statement 
for $f_{(j)}, k_{(j)a}$ for increasing $j=1,2, \dots$. For this, we 
must use the EoMs of our EFT, as well as our conditions for the characteristic surface $\cH$ in our EFT. 

These conditions are best analyzed by decomposing the metric, and its various expansion orders in $\ell$, as in \eqref{gGNC}, 
with corresponding tensors $p^a{}_b, q^{ab}, l^a, h_{ab}, k^a = q^{ab}k_b$ and functions (coordinates) $u,v$, all of which are Lie-derived by $K^a$ and by $\phi^a$, see sec. \ref{sec:KH} and app. \ref{app:GNC}. 
This analysis is rather involved and the full technical details are therefore moved to app. \ref{app:proofA}. 

A series of lemmas first leads to the conclusion that $k_{(1)}^a|_{\cH} = \omega_{(1)} \phi^a|_{\cH}$, where $\omega_{(1)}$ is constant on $\cH$, i.e. $k_{(1)}^a$ is an infinitesimal rotation on $\cH$. It is also shown that 
\be
f_{(1)}|_{\cH} = f_{(2)}|_{\cH} = 0,
\ee
and that $(f+k^2)_{(3)}|_{\cH}=0$.
Based on these insights, we redefine 
\be
K^a \to K^a - \ell \omega_{(1)} \phi^a, 
\ee
which still Lie-derives $(\Phi,g_{ab})$. Having made this redefinition,  we set up our gGNC system as in \eqref{gGNC} w.r.t. this new $K^a$. The corresponding new $k_a, f$ are seen to satisfy all previous properties, and additionally 
\be
f_{(3)}|_{\cH} = k_{(1)a}|_{\cH} = 0. 
\ee
The remainder of the proof repeats the arguments for increasing orders in $\ell$, assuming inductively that for some $n\ge 0$
\begin{subequations}
\begin{align}
    f_{(j)}|_{\cH} &=0 \text{ for all } j\leq n+3 \, , \\
    k_{(j)a}|_{\cH} &= 0 \text{ for all } j\leq n+1 \, .
\end{align}
\end{subequations}
By similar arguments as in the base case, these are seen to imply
that $(f+k^2)_{(n+4)}|_{\cH}=0$ and that $k_{(n+2)}^a|_{\cH} = \omega_{(n+2)} \phi^a|_{\cH}$, where $\omega_{(n+2)}$ is constant on $\cH$, i.e., it is proportional to a rotation.

Similarly as in the base case, we redefine 
\be
K^a \to K^a - \ell^{n+2} \omega_{(n+2)} \phi^a, 
\ee
which still Lie-derives $(\Phi,g_{ab})$. Having made this redefinition,  we set up our gGNC system as in \eqref{gGNC} w.r.t. this new $K^a$. The corresponding new $k_a, f$ are seen to satisfy all the previous properties, and additionally 
\be
f_{(n+4)}|_{\cH} = k_{(n+2)a}|_{\cH} = 0, 
\ee
thereby closing the induction loop. 

After our inductive re-adjustments, the final Killing field $K^a$ and gGNCs are by construction such that $K_a|_{\cH} = (\ud u)_a |_{\cH}$ at points of $\cH$, i.e., it is characteristic, and 
$\cH$ is a Killing horizon w.r.t. $K^a$. Strictly speaking, 
we still need to show that the series defining $K^a$ converges. This follows from analyticity in a fairly straightforward manner, but we will not go through the somewhat tedious details here.
\qed

\section{Theorem B}
\label{sec:thmB}

\subsection{Assumptions}
\label{sec:BAss}
For the formulation of Theorem B, we introduce the \emph{domain of outer communication} $\sD = I^+(\sI^-) \cap I^-(\sI^+)$ w.r.t. the propagation cones of our EFT (sec. \ref{sec:propcones}). We then define the future and past horizons as 
$\sH^\pm = \partial \sD \cap I^\pm(\sI^\mp)$, as usual, but w.r.t. the propagation cones of our EFT.

We assume the following: 
    \begin{enumerate}
        \item[(i')] $(\cM(\ell), \Phi(\ell,x),g_{ab}(\ell,x))$ is a 1-parameter family of solutions to the EFT \eqref{lagrangian1} that is jointly real analytic in $(x,\ell)$ under some 1-parameter family of diffeomorphisms identifying each $\cM(\ell)$ with some fixed $\cM$.
        \item[(ii')] Each $\cM(\ell)$ is asymptotically Minkowskian and there is a vector field $t^a(\ell)$ that Lie-derives $(\Phi(\ell), g_{ab}(\ell))$, that is jointly analytic in $(x,\ell)$, that has complete orbits, and that is timelike in the asymptotic region. We require that $\sM(\ell)$ includes an open neighborhood of $\cH(\ell)$.
        \item[(iii')] 
For $\ell = 0$, $\sH^+$ forms an non-degenerate $(\kappa >0)$ Killing horizon. The pull-back of $g_{ab}(\ell)$ to $\sH^+(\ell)$ has an isometry $\iota$ fixing some cut $H(\ell)$, such that\footnote{$\iota$ could depend discontinuously on $\ell$.} $\iota_* t^a(\ell) = -t^a(\ell)$, $\iota^2=id$. The complex structure of the cut $H(\ell)$ satisfies 
        \be
        \label{ioch}
\iota^* J^a{}_b = -J^a{}_b.
        \ee
        \item[(iv')] 
For the $\ell=0$ solution, we have
\be
V(\Phi) \ge 0
\ee
on $\cH$.
        \item[(v')] 
        The restriction of $t^a$ tangent to $\sB$ in the $\ell=0$ solution is not vanishing identically, i.e. the black hole is rotating.
    \end{enumerate}

As we explain in more detail in app. \ref{app:iota}, 
(iii') states in a covariant way that $\iota$ is a kind of "$t$-$\varphi$" reflection isometry on $\sH^+$, without actually introducing\footnote{Note that we do not know at this stage that there is a Killing field "$\phi^a = (\partial_\varphi)^a$".} coordinates $t$ or $\varphi$.
Note that we only require the symmetry $\iota$ locally, i.e., for the pull-back to $\cH$. This geometric feature is present in Kerr, but the reflection $\iota$ on $\cH$ is different from the global $t$-$\varphi$ reflection isometry of Kerr
which exchanges $\sH^\pm$. Assuming such a global symmetry would enable a "purely kinematical" proof similar to that of prop. \ref{prop1}.


To picture the hypothetical case that $\cH$ might not be a Killing horizon,
we may imagine that $\sH^+$ as defined by the Garding cone is a time-like surface with respect to $g_{ab}$ tangent to a boost-like Killing field $K^a$ similar to fig. \ref{fig:ABhorizons}. The future pointing Garding cones must then tilt inwards relative to those of $g_{ab}$ everywhere along $\sH^+$, so that $\sH^-=\emptyset$. We note that there is an evident asymmetry between past and future, which is already somewhat at odds with thermodynamic intuition.
      
Apart from the global assumptions such as the reflection isometry and asymptotic flatness, or topology of $\cH$, the main differences between these assumptions for Theorem B and those formulated for Theorem A in sec. \ref{sec:thmA} are: We do not require, a priori, any rotational vector field  that Lie-derives $(\Phi, g_{ab})$, and we do not require that the bicharacteristics are tangent to a vector field that Lie-derives $(\Phi, g_{ab})$. 
On the other hand, we need the stability requirement $V\ge 0$.

If we were to assume instead of (iv') that the $\ell=0$ black hole solution is \emph{non-rotating}, then the solution would be static \cite{Sudarsky:1993kh}. As shown in the proofs of the uniqueness of static black hole solutions for a wide class of Einstein-scalar field theories, see e.g., \cite{Gibbons:2002av,bunting2,rogatko2}, the $\ell=0$ solution $(\cM, \Phi, g_{ab})$ would in fact have to be \emph{spherically symmetric}. An expansion argument in $\ell$ would then show that this remains true for $\ell \neq 0$ \cite[sec. 4]{Hollands:2022ajj}. We will not treat explicitly this comparatively easier case, but see e.g., \cite{Benkel:2018qmh} for corresponding arguments, leading to the same conclusion as our Theorem B in the spherically symmetric case.

\subsection{Statement}
\label{sec:statement}
Theorem B is the following statement:

\begin{theorem}
    \label{thmB}
Under the assumptions (i')--(v'), $\sH^+$ is an  compactly generated Killing horizon w.r.t. the metric $g_{ab}$.
\end{theorem}

{\it Remark:} 
(1) Again, we also have cor. \ref{cor:1}, and 
the proof also shows that the symmetry group of $(\Phi,g_{ab})$ is (at least)
${\mathbb R} \times U(1)$.

(2) If we make the assumptions  fully symmetric in past and future, then we analogously get a \emph{bifurcate} Killing horizon.

(3) By (iv'), 
the scalar field stress energy tensor of the $\ell=0$ theory satisfies the dominant energy condition on $\cH$. By Hawking's topology theorem 
\cite{hawking}, therefore $\sB \cong {\mathbb S}^2$. 

\subsection{Proof}
\label{sec:thmbproof}

We argue in a broadly similar way as in the proof of Theorem A. First, we get a second continuous symmetry of $(\Phi,g_{ab})$.
For $\ell = 0$, $t^a$ is tangent to $\sB = \partial \sH^+$, i.e., its restriction to $\sB$ must be $\Omega \phi^a$, where $\phi^a$ has $2\pi$-periodic orbits. By the arguments of \cite{Hollands:2022ajj}, this $\phi^a$ extends in view of (i')---(v'), to a rotational vector field $\phi^a(\ell)$ Lie-deriving $(\Phi(\ell),g_{ab}(\ell))$ on the entire spacetime for $|\ell|<\ell_0$ and a sufficiently small but positive $\ell_0$.\footnote{Ref. \cite{Hollands:2022ajj} considered the horizons defined by the lightcones of the metric $g_{ab}$, not the Garding cones. Of course we are free also in our situation to consider the metric lightcones to make the argument, since the Killing field extends to as a single-valued object (due to $\pi_1(\sD) = 0$ \cite{Friedman:1993ty}) to all of $\sD$ by analyticity.}. 

Next, we introduce gGNCs \eqref{gGNC1} around $\cH(\ell)$,  based on the cut  $H(\ell)$ as in (iii'), initially taking $K^a = t^a + \Omega_{(0)} \phi^a$. [As in the proof of theorem A, this $K^a$ is subsequently updated.] Again, we identify points of $\sM(\ell)$ for different values of $\ell$ if these gGNCs coincide. Thereby we get the corresponding tensor fields $l^a, K^a, \phi^a, p^a{}_b, k_a, q^{ab}, h_{ab}$, see \eqref{pqdef} in sec. \ref{sec:gGNC}, and  app. \ref{app:GNC} for details. After our identification, $K^a, p^a{}_b, t^a, \phi^a, H, \sH^+$ are independent of $\ell$. 

As in the proof of Theorem A, we next make power series expansions of all quantities as in, e.g. \eqref{power_series}, with the aim of showing $f=k_a=0$ on $\cH$---this implies that $K^a$ is tangent and normal to $\cH$, i.e., it is a Killing horizon. 

Since for $\ell = 0$, we have an ordinary Einstein-scalar field theory, the results by \cite{Hollands:2006rj} show that $\cH$ is a Killing horizon w.r.t. $K^a$ with constant surface gravity $\kappa = \frac{1}{2} \sL_l f_{(0)}$ and constant $\Phi_{(0)}$. Thus, we have $D_a \Phi^{(0)} = f_{(0)} = k_{(0)a} = K_{(0)ab} = 0$ on $\cH$. 

Analogously to the proof of Theorem A in app. \ref{app:proofA} we can also evaluate the condition $Q^{abcd} (\ud u)_a (\ud u)_b (\ud u)_c (\ud u)_d=0$, expressing that $\cH$ is characteristic, 
at subsequent orders in $\ell$.
At order $\ell^2$, \eqref{Q1} first shows that $(f + k^2)_{(1)}= 0$ on $\cH$. Since we already know that 
$f_{(0)}=k_{(0)a}=0$ on $\cH$, we get $f_{(1)}=0$ on $\cH$.

Consider a change $v \to v + \lambda(x^A)$, i.e., we change the cut $H$ of $\cH$ by an angle-dependent amount. The gGNC associated with the new cut is such that
\begin{equation}
    \sL_l k_{(0)a} \to \sL_l k_{(0)a} -2\kappa D_a \lambda
\end{equation}
There is a real analytic solution $\lambda$ to
\begin{equation}
    D^a D_a \lambda = \frac{1}{2\kappa} D^a \sL_l k_{(0)a} \quad 
    \text{on $H$,}
\end{equation}
by standard results about elliptic partial differential equations, see e.g., \cite{evans}. Here and in the rest of the proof, $D_a$ temporarily denotes the angular
covariant derivative w.r.t. $h_{(0)ab}$. In the new real analytic gGNC, we have
\begin{equation}
\label{kdiv}
    D^a \sL_l k_{(0)a} = 0 
    \, 
    \Longrightarrow 
    \,
    \sL_l k_{(0)a} = \epsilon_{ab} D^b F_{(1)} 
\end{equation}
for some $F_{(1)}$ for which $\sL_\phi F_{(1)} = \sL_K F_{(1)} = 0$ on $\cH$. This implies that
\begin{equation}
\label{kdiv2}
(D_a \sL_l k_{(0)b}) \phi^a \phi^b = 0
\end{equation}
on $\cH$. We also know that $k^a_{(1)} = \omega_{(1)} \phi^a$ for some function $\omega_{(1)}$ since we have $\iota$ in (iii'), see app. \ref{app:iota}.

Our aim is to show that $\omega_{(1)}$ is constant on $\cH$.
For this, we first transvect the EoM \eqref{EoMg} of our EFT for the metric into $K^a p^b{}_c$. Evaluating the result at order $\ell$ yields
\begin{align}
\label{R1vA}
R_{(1)ab} K^a p^b{}_c = \frac{1}{2} k_{(1)c} V(\Phi_{(0)}).
\end{align}
For the Ricci-tensor component on the left side we employ the expressions in app. \ref{app:GNC} at order $\ell$, which yields
\begin{widetext}
\begin{align}
\begin{split}
\label{ddk}
0=&D^b D_{[a} k_{(1)b]} - \left[2\kappa \bar K_{(0)ab} +\frac{1}{2} h_{(0)ab}V(\Phi_{(0)}) + \frac{1}{2} \sL_l k_{(0)a} \sL_l k_{(0)b} \right] k_{(1)}^b \\
&+ \sL_l k_{(0)[a} D_{b]} k_{(1)}^b
- \frac{1}{2} D_a \sL_l f_{(1)}
+  \left[ 
D_b \sL_l k_{(0)a} - \frac{1}{2} D_a \sL_l k_{(0)b}
\right]k_{(1)}^b 
\end{split}
\end{align}
\end{widetext}
on $\cH$. Here we use the shorthand $k^a = q^{ab}k_b$, disregarding temporarily our usual 
convention that indices are 
raised with $g^{ab}$. 

Next, we transvect \eqref{ddk} with $\phi^a$ using 
\eqref{kdiv2} and $k^a_{(1)} = \omega_{(1)} \phi^a$.
In terms of the uniformizing angular coordinates $(\vartheta,\varphi)$
on $H \cong {\mathbb S}^2$ [see eq. \eqref{gittel}], setting $\xi=\cos \vartheta$ and $y=k_{(1)a} \phi^a$, we obtain the ordinary differential equation
\begin{equation}
\label{yeq}
\left[ 
\frac{\ud}{\ud \xi} \left(\Psi^{-1} \frac{\ud}{\ud \xi} \right) + W
\right] y = 0,
\end{equation}
where 
\begin{equation}
W = \frac{
4\kappa \langle  \phi \otimes  \phi, \bar K_{(0)} \rangle + V(\Phi_{(0)}) \| \phi \|^2 +  \langle \phi, \sL_l k_{(0)} \rangle^2
}{\|\phi\|^2}    
\end{equation}
Here we use the compact notations $\langle \, . \, , \, . \, \rangle, \| \, . \, \|$ for angular tensor contractions introduced in \eqref{lrdef} and \eqref{normdef} (applied to $q_{(0)}^{ab}$).

The function $W=W(\xi)$ is regular on $[-1,1]$ up to and including the boundary, as is $\Psi^{-1}= \Psi(\xi)^{-1}$ [see \eqref{gittel}]. $y=y(\xi)$ is regular and vanishes at the boundary. There can be at most one such function up to rescaling, and this function is, in fact, given by $y=\|\phi \|^2$. To see this, note that if we were to change $K^a \to K^a - \ell \omega_{(1)} \phi^a$ for a \emph{constant} $\omega_{(1)}$ and then used this new $K^a$ to construct or gGNCs, then 
in these coordinates $k_{(1)}^a$ would be replaced by $k_{(1)}^a-\omega_{(1)} \phi^a$, see lem. \ref{lem8}. 
Thus, $\omega_{(1)} \phi^a$ with constant $\omega_{(1)}$
solves \eqref{ddk}. This must in fact be the only solution.
since any other solution $y=k_{(1)a} \phi^a$ to \eqref{yeq}
would be non-vanishing at least at one boundary point 
$\xi=\pm 1$, corresponding to a $k_{(1)}^a$ that is irregular at the north/south pole of the sphere $H$.

As in the proof of Theorem A in sec. \ref{app:proofA},
we next update our initial guess of $K^a$ to $K^a \to K^a - \ell \omega_{(1)} \phi^a$, where $\omega_{(1)}$ is the constant such that $k_{(1)}^a = \omega_{(1)}\phi^a$. 
In the gGNCs based on the new $K^a$, we then have 
$k_{(0)}^a = k^a_{(1)} = f_{(0)} = f_{(1)} = 0$ on $\cH$.

The next step is to evaluate the condition \eqref{Q1} at order $\ell^4$, giving that $(f + k^2)_{(2)} = 0$ on $\cH$. Thus, we additionally get $f_{(2)}=0$ on $\cH$. Then furthermore $(f + k^2)_{(3)} = 0$ by \eqref{Q11}. 

This set of arguments is now repeated in an inductive manner similar to app. \ref{sec:indstep} for Theorem A. We assume that, on $\cH$,  
$f_{(j)}=0$ for all $j \le n+1$, $k_{(j)a}=0$ for all 
$j \le n$, and $(f+k^2)_{(n+2)}=0$. 
Again, we transvect the EoM \eqref{EoMg} of our EFT for the metric into $K^a p^b{}_c$, and now evaluate the result at $O(\ell^{n+1})$. This yields \eqref{R1vA}
with $k_{(1)a} \to k_{(n+1)a}$. The reason for this is the following: the terms with explicit powers of $\ell$ in \eqref{EoMg} all vanish at this order by the induction hypothesis, as a consequence of the fact\footnote{The proof of this well-known result is of the same nature as the proof of prop. \ref{prop1}.}  
that a contraction $p^a{}_b K^c T_{bc}$ of a local covariant tensor $T_{ab}$ built from $(\Phi, g_{ab})$
vanishes on a Killing horizon w.r.t. $K^a$.

Proceeding next along the same lines as before gives $k_{(n+1)a} = \omega_{(n+1)}\phi^a$ for a constant $\omega_{(n+1)}$, which we use to update
$K^a \to K^a - \ell^{n+1} \omega_{(n+1)}\phi^a$. In the new gGNC based on this new $K^a$, we additionally have $k_{(n+1)a}=0$ on $\cH$. When combined with 
$(f+k^2)_{(n+2)}=0$ and the induction hypothesis, we 
get $f_{(n+2)} = 0$ on $\cH$. Finally, 
\eqref{Q11} which holds at this stage with the replacements $k_{(1)a} \to k_{(n+1)a}$ and 
$(f+k^2)_{(3)} \to (f+k^2)_{(n+3)}$
 then tells us that $(f + k^2)_{(n+3)} = 0$, thus closing the induction loop. \qed

\section{Outlook}

While our complementary Theorems A and B cover a fairly general class of EFTs and solutions, it would be interesting to remove some of their technical assumptions. Concerning theories, one should try to include, at least at the derivative order in the Lagrangian considered here, multiple scalar fields and abelian vector fields while staying within the class of theories with second order EoMs. 

Concerning the assumptions in Theorem B, the reflection symmetry condition at $\sH^+$ makes intuitive sense for equilibrium solutions in theories as our EFTs that do not introduce an explicit time direction such as Einstein-aether theories\footnote{In fact, \cite{Adam:2021vsk} have found stationary black hole solutions in such theories which do not have this reflection symmetry, and which do not have Killing horizons.} \cite{Adam:2021vsk}. The other assumption, $V \ge 0$, is also well-motivated as a stability condition and, e.g., required in other standard theorems in general relativity, such as the positive mass theorem.

Concerning Theorem A, the most restrictive assumption is that the Killing orbits line up with the bicharacteristics of the propagation cone. To remove this, one would evidently have to bring in the global structure of spacetime, which is not considered in Theorem A as all its assumptions are of a local nature.

\medskip
\noindent
{\bf Acknowledgements:} SH thanks Harvey Reall and Thomas Sotiriou for helpful discussions and comments, and Trinity College, Cambridge University, for support.

\appendix{

\section{Curvature components in gGNCs}\label{app:GNC}

Below we provide the expressions for the curvature components 
of the following metric given by
\be
\label{gGNC1}
g=2\ud v\ud u-f \ud v^2+2k_A \ud x^A \ud v +h_{AB}\ud x^A \ud x^B
\ee
(For equivalent expressions in a different notation, see \cite{Hollands:2022fkn}.)
The nonzero components of the inverse metric are
\be
g^{uu}=(f+k^2), \; g^{uv}=1, \; g^{uA}=-k^A, \; g^{AB}=h^{AB}
\ee
where $h^{AB}$ is the inverse of $h_{AB}$, $k^A\equiv h^{AB}k_B$ and $k^2\equiv k_A k^A$. We do not assume in this section that $k_A=f=0$
for $u=0$, as would be the case for the usually defined GNCs off of the null surface $\cN$ defined by $u=0$. 

We introduce $l^a = (\partial_u)^a$, $K^a = (\partial_v)^a$, $p^a{}_b, h_{ab}, D_a$ as in 
\eqref{pqdef}, \eqref{Ddef}, \eqref{hdef}, and
the following variables
\be
\label{Kabdef}
K_{ab}\equiv \frac12 \sL_K h_{ab}, \qquad \bar K_{ab}\equiv \frac12 \sL_l h_{ab}, 
\ee
and 
\be
K = q^{ab} K_{ab}, \; \bar K = q^{ab} \bar K_{ab}.
\ee
In this appendix, $K^a$ is not assumed to be a Killing vector field, though in the main part of the paper, it is, resulting in simplified expressions obtained by dropping $\sL_K k_a, \sL_Kf, K_{ab}, K$ everywhere.

The nonzero Christoffel symbols are given by
\begin{widetext}
\begin{subequations}
\begin{align}
    K^al^b\nabla_a(\ud u)_b&=-\frac{1}{2} \sL_l f- \frac{1}{2}q^{ab}k_b\sL_l k_a \, , \\
    K^al^b (\partial_A)^c \nabla_a(\ud x^A)_b&= \frac{1}{2} q^{ca}\sL_lk_a \, , 
\end{align}
\begin{align}
    l^ap^b{}_c\nabla_a(\ud u)_b &= \frac{1}{2} \sL_lk_c - q^{ab}k_b \bar{K}_{ca} \, , \\
    l^ap^b{}_c (\partial_A)^d\nabla_a(\ud  x^A)_b &= \bar{K}^d{}_c \, , \vphantom{\frac{a}{b}} \\
    K^aK^b\nabla_a(\ud v)_b &= \frac{1}{2} \sL_l f \, , \\
    K^aK^b\nabla_a(\ud u)_b &= - \frac{1}{2} \sL_K f + \frac{1}{2} (f+k^2) \sL_l f - q^{ab}k_b \left( \sL_K k_a - \frac{1}{2} D_a f \right) \, , \\
    K^aK^b (\partial_A)^c\nabla_a (\ud x^A)_b &= \frac{1}{2} q^{ca} k_a \sL_l f 
    + q^{ca} \left( \sL_K k_a + \frac{1}{2} D_af \right) \, , \\
    K^ap^b{}_c \nabla_a(\ud v)_b &= -\frac{1}{2} \sL_l k_c \, , \\
    K^ap^b{}_c\nabla_a(\ud u)_b &= - \frac{1}{2} (f+k^2) \sL_lk_c - \frac{1}{2}D_cf - q^{ab}k_bD_{[c}k_{a]} - q^{ab}k_b K_{ca} \, , \\
    K^ap^b{}_c (\partial_A)^d \nabla_a(\ud x^A)_b  &= \frac{1}{2} q^{da}k_a \sL_l k_c - q^{da}D_{[a}k_{c]} + q^{da} K_{ac} \, , \\
    p^a{}_bp^c{}_d \nabla_a (\ud v)_c &= - \bar{K}_{bd} \, , \vphantom{\frac{a}{b}} \\
    p^a{}_bp^c{}_d\nabla_a (\ud u)_c &= - (f+k^2) \bar{K}_{bd } + D_{(b}k_{d)} - K_{bd} \, , \vphantom{\frac{a}{b}} \\
    p^a{}_b p^c{}_d (\partial_A)^e \nabla_a (\ud x^A)_c &= p^a{}_b p^c{}_d (\partial_A)^e D_a (\ud x^A)_c + q^{ea}k_a \bar{K}_{bd} \, . \vphantom{\frac{a}{b}}
\end{align}
\end{subequations}
The components of the Riemann tensor are as follows:
\begin{subequations}
\begin{align}
    l^aK^bl^cK^dR_{abcd} &= \frac{1}{4} q^{ab}\sL_l k_a\sL_lk_b + \frac{1}{2} \sL_l\sL_l f \, , \\
    l^aK^bl^cp^d{}_eR_{abcd} &= \frac{1}{2} q^{ab}  \bar{K}_{eb} \sL_lk_a - \frac{1}{2} \sL_l\sL_l k_e \, , \\
    l^aK^bK^cp^d{}_e R_{abcd} &= - \frac{1}{2} \sL_K\sL_l k_e - \frac{1}{2} q^{ab} K_{ea} \sL_lk_b + q^{ab} \bar{K}_{ea} \sL_K k_b - \frac{1}{4} q^{ab}k_a \sL_lk_e \sL_lk_b - \frac{1}{2} q^{ab} \sL_l k_b D_{[e}k_{a]} \notag \\ 
    &\phantom{=}- \frac{1}{2} D_e\sL_l f + \frac{1}{2}q^{ab} \bar{K}_{eb}D_af - \frac{1}{2} q^{ab}k_b \bar{K}_{ea} \sL_l f \\
    l^ap^b{}_cl^dp^e{}_fR_{abde} &= q^{ab} \bar{K}_{cb} \bar{K}_{fa} - \sL_l \bar{K}_{cf} \vphantom{\frac{a}{b}} \\
    l^ap^b{}_cK^dp^e{}_f R_{abde} &= - \sL_K \bar{K}_{cf} + q^{ab} K_{cb} \bar{K}_{fa} + q^{ab} k_b \bar{K}_{f[a} \sL_l k_{c]} + q^{ab} \bar{K}_{fb} D_{[c}k_{a]} + \frac{1}{2} D_f\sL_l k_c - \frac{1}{4} \sL_l k_c\sL_lk_f \notag \\ 
    &\phantom{=}- \frac{1}{2} \bar{K}_{cf} \sL_l f \, ,\\
    K^ap^b{}_cK^dp^e{}_f R_{abde} &= - \sL_K K_{cf} + q^{ab} K_{cb} K_{fa} - 2q^{ab}k_b \bar{K}_{c[f} \sL_l k_{a]} - 2q  ^{ab}k_b \bar{K}_{f[c} \sL_l k_{a]}  +\frac{1}{4} (f+k^2) \sL_lk_c\sL_lk_f \vphantom{\frac{a}{b}} \notag \\ 
    &\phantom{=}+ \frac{1}{2} \sL_lk_{(c}D_{f)}f - q^{ab} k_b \sL_lk_{(c}D_{f)}k_a  - q^{ab}K_{a(c|}D_bk_{|f)} + q^{ab} K_{a(c}D_{f)}k_b \vphantom{\frac{a}{b}} + D_{(f} \sL_K k_{c)} \notag \\ 
    &\phantom{=}+ \frac{1}{2}D_fD_c f -\frac{1}{2} \bar{K}_{cf} q^{ab}k_bD_af + \frac{1}{2} q^{ab}k_b \sL_l k_{(c|}D_ak_{|f)} +q^{ab}D_{[a}k_{c]} D_{[b}k_{f]} + \frac{1}{2} K_{cf}\sL_l f \notag \\ 
    &\phantom{=}+ \frac{1}{2}(f+k^2) \bar{K}_{cf} \sL_lf + \frac{1}{2} D_{(c}k_{f)} \sL_l f - \frac{1}{2} \bar{K}_{cf} \sL_K f \, , \\
    p^a{}_bp^c{}_dK^ep^f{}_g R_{acef} &= -2D_{[b}K_{d]g} + \sL_l k_{[b}K_{d]g} - 2q^{ac} K_{a[b}\bar{K}_{d]g}k_c - \sL_l k_{[b} \bar{K}_{d]g} (f+k^2)  \vphantom{\frac{a}{b}} + \bar{K}_{g[b}D_{d]}f \notag \\ 
    &\phantom{=}+ \frac{1}{2} \sL_lk_{[b}D_{d]} k_g - \frac{1}{2} D_{g}k_{[b} \sL_lk_{d]} -D_gD_{[b}k_{d]} +  q^{ac}\bar{K}_{g[b|} k_cD_{|d]}k_a - q^{ac}k_a \bar{K}_{g[b|}D_ck_{|d]} \, , \vphantom{\frac{a}{b}} \\
    p^a{}_bp^{c}{}_dl^ep^f{}_gR_{acef} &= 2D_{[d} \bar{K}_{b]g} - 2q^{ac}\bar{K}_{a[b} \bar{K}_{d]g}k_c + \sL_lk_{[b} \bar{K}_{d]g} \, , \vphantom{\frac{a}{b}} \\
    p^a{}_bp^c{}_dp^e{}_fp^g{}_hR_{aceg} &= R[h]_{bdfh} + 2K_{d[f} \bar{K}_{h]b} + 2K_{b[h}\bar{K}_{f]d} + 2(f+k^2) \bar{K}_{b[h} \bar{K}_{f]d} \notag \vphantom{\frac{a}{b}} - \bar{K}_{h[d}D_{b]}k_f - \bar{K}_{f[d}D_{b]}k_h \notag \\ 
    &\phantom{=}+ \bar{K}_{d[h}D_{f]}k_b - \bar{K}_{b[h}D_{f]}k_d \, . \vphantom{\frac{a}{b}}
\end{align}
\end{subequations}
The components of the Ricci tensor are given by the following expressions
\begin{subequations}
\begin{align}
    K^aK^bR_{ab} &= - q^{ab} \sL_K K_{ab} + q^{ac}q^{bd} K_{ab}K_{cd} - \bar{K}q^{ab}k_b \sL_lk_a + 2 q^{ac}q^{bd} \bar{K}_{cd}k_a \sL_l k_b + q^{ac}q^{bd} \bar{K}_{ab}k_cD_df - \frac{1}{2} \bar{K} q^{ab}k_bD_af \vphantom{\frac{a}{b}}  \notag \\
    &\phantom{=}- q^{ab}q^{cd}k_b \sL_lk_cD_ak_d \vphantom{\frac{a}{b}} + q^{ab}D_b\sL_Kk_a - q^{ab}k_bD_a \sL_lf + q^{ab}q^{cd}k_b \sL_l k_cD_dk_a + \frac{1}{2} q^{ab}D_{b}D_af - \frac{1}{2} \bar{K} \sL_Kf \notag \\
    &\phantom{=} - q^{ac} q^{bd} D_{[a}k_{b]}D_dk_c + \frac{1}{2} K\sL_lf + \frac{1}{2} (f+k^2) \bar{K} \sL_l f - q^{ab}k_b\sL_K\sL_lk_a  - \bar{K}_{ab}k^ak^b \sL_lf + \frac{1}{2} q^{ab}D_ak_b \sL_l f \notag \\ 
    &\phantom{=}+ \frac{1}{2}(f+k^2) \sL_l\sL_l f + \frac{1}{2} (f+k^2)q^{ab} \sL_lk_a \sL_l k_b - \frac{1}{2} q^{ab}q^{cd} k_a k_c \sL_l k_b \sL_l k_d + \frac{1}{2} q^{ab} \sL_l k_aD_bf \, , \\
    l^aK^bR_{ab} &= - q^{ab} \sL_K \bar{K}_{ab} + q^{ac} q^{bd} K_{cd} \bar{K}_{ab} - \frac{1}{2} \bar{K} q^{ab} k_b \sL_l k_a + q^{ac}q^{bd} \bar{K}_{cd} k_a \sL_l k_b - \frac{1}{2} \sL_l\sL_l f - \frac{1}{2} q^{ab} \sL_lk_a\sL_lk_b \notag \\ 
    &\phantom{=}+ \frac{1}{2} q^{ab}D_b\sL_lk_a - \frac{1}{2} \bar{K} \sL_l f - \frac{1}{2}q^{ab}k_b \sL_l\sL_lk_a \, , \\
    l^al^bR_{ab} &= q^{ac} q^{bd} \bar{K}_{ab}\bar{K}_{cd} - q^{ab} \sL_l \bar{K}_{ab} \, , \vphantom{\frac{a}{b}} \\
    l^ap^b{}_cR_{ab} &= q^{ab} D_b\bar{K}_{ca} - D_c\bar{K} + \frac{1}{2} \sL_l\sL_lk_c + \frac{1}{2} \bar{K} \sL_l k_c - q^{ab}\bar{K}_{cb} \sL_lk_a \notag \\
    &\phantom{=} + q^{ab}k_b \left( 2 q^{de} \bar{K}_{ce} \bar{K}_{ad} - \bar{K}_{ca} \bar{K} - \sL_l\bar{K}_{ca} \right) \vphantom{\frac{a}{b}} \, , \\
    K^ap^b{}_cR_{ab} &= q^{ab} D_a K_{cb} - D_cK - \frac{1}{2}\sL_K\sL_l k_c - \frac{1}{2}K \sL_lk_c + q^{ab} \bar{K}_{cb} \sL_Kk_a + \sL_lk_{[c}\bar{K}_{a]b}k_dk_e q^{ad}q^{be} - \frac{1}{2} q^{ab} k_b \sL_l k_c\sL_l k_a \notag \\
    &\phantom{=} + (f+k^2) q^{ab} \bar{K}_{cb} \sL_l k_a - \bar{K} q^{ab} k_bD_{[c}k_{a]} +  q^{ab}\sL_l k_{[c}D_{a]}k_b  + 2 q^{ad}q^{be} \bar{K}_{b[c|}k_dD_ek_{|a]} - 2q^{ad}q^{be} \bar{K}_{b[c}k_dD_{a]}k_e \vphantom{\frac{a}{b}} \notag \\ 
    &\phantom{=}+ 2 q^{ad}q^{be} K_{ce} \bar{K}_{ab} k_d - q^{ab} K_{ca} \bar{K}k_b  - \frac{1}{2}D_c\sL_l f + q^{ab} D_bD_{[c}k_{a]} + q^{ab} \bar{K}_{ca}D_bf - \frac{1}{2}\bar{K}D_c f \notag \\
    &\phantom{=} +  q^{ab}k_b D_a\sL_lk_c - \frac{1}{2} q^{ab} k_b D_c \sL_lk_a  - q^{ab} \bar{K}_{ca}k_b \sL_l f   - \frac{1}{2} (f+k^2) \left( \sL_l\sL_lk_c + \bar{K} \sL_l k_c \right) - q^{ab} k_b \sL_K \bar{K}_{ca} \, ,\\
    p^a{}_bp^c{}_dR_{ac} &= R[h]_{bd} - 2\sL_K \bar{K}_{bd} - K \bar{K}_{bd} + 4 q^{ac} K_{(b|c} \bar{K}_{|d)a} - K_{bd} \bar{K} + (f+k^2) \left( 2q^{ac}\bar{K}_{bc} \bar{K}_{da} - \bar{K}_{bd} \bar{K} - \sL_l \bar{K}_{bd} \right) \vphantom{\frac{a}{b}} \notag \\
    &\phantom{=} - 4 q^{ad}q^{ce} \bar{K}_{b[a} \bar{K}_{d]c}k_dk_e + 2 q^{ac} \sL_l k_{(b} \bar{K}_{d)a}k_c - 2 q^{ac} \bar{K}_{bd} k_c \sL_l k_a - \bar{K}_{bd} \sL_l f  -2 q^{ac} k_cD_{[b}\bar{K}_{a]d}  \vphantom{\frac{a}{b}} \notag \\
    &\phantom{=} - 2 q^{ac} k_c D_{[d}\bar{K}_{a]b} +  \bar{K}D_{(d}k_{b)}  + \frac{1}{2} D_{(d} \sL_lk_{b)} - \frac{1}{2} \sL_l k_b\sL_lk_d + \bar{K}_{bd} q^{ac} D_ak_c - 2 q^{ac} \bar{K}_{a(b|}D_ck_{|d)} \, . \vphantom{\frac{a}{b}}
\end{align}
\end{subequations}
\end{widetext}

\section{Big-$O$ Notation}

In the proofs of Theorem A and Theorem B, we will have to analyze 
whether certain tensor fields vanish on a Killing horizon, $\cH$, w.r.t. a Killing vector field $K^a$.
These tensors will be built from $q^{ab}, h_{ab}, f, k_a$, their Lie-
along $K^a, l^a$, and their angular covariant derivatives $D_a$, where these quantities refer to gGNCs, see sec. \ref{sec:KH}.
On a Killing horizon $\cH$, we have $f, k_a=0$, so also their $D_a$ and $\sL_K$ derivatives vanish as these point along $\cH$. On the other hand, the Lie derivative $\sL_l$ need not vanish. In the proofs of Theorem A and Theorem B, we do not a priori know that $f, k_a=0$, and that their $D_a$ and $\sL_K$ derivatives vanish at $\cH$, because we precisely want to prove that $\cH$ actually is a Killing horizon. 

In order to retain control of the order (in the gGNC $u$ defining $\sH$ through $u=0$) at which certain expressions vanish in various perturbative arguments, we must keep track of the powers of $k_a, f$ and their 
their $D_a$ and $\sL_K$ derivatives. 

For this, we will now introduce a convenient notation:
If the power of $f$ is $r$ and that of $k_a$ is $s$, we write that such an expression is $O(f^r k^s)$. 
More precisely, a tensor $T^{a \dots b}{}_{c \dots d}$ built from $q^{ab}, h_{ab}, f, k_a, K^a, l^a$, their Lie-derivatives along $K^a, l^a$, and their angular covariant derivatives $D_a$, is said to be $T=O(f^r k^s)$, if it is a linear combination of simultaneous eigenvectors of eigenvalues $(\ge r,\ge s)$
of the operators
\begin{align}
N_f =& \sum_{p,q} \sL_K^p D_{(a_1 \dots a_q)}^q f 
\frac{\partial}{\partial(\sL_K^p D_{(a_1 \dots a_q)}^q f)},\\
N_k =& \sum_{p,q} \sL_K^p D_{(a_1 \dots a_q)}^q k_c 
\frac{\partial}{\partial(\sL_K^p D_{(a_1 \dots a_q)}^q k_c)}.
\end{align}
We also write, e.g., $T = O(kf,k^2)$ if it is a sum of terms of the order $O(k^2)$ and $O(kf)$. As examples:
\be
q^{cb} \sL_K D_cf D_a k_b = O(fk), 
\ee
but 
\be
q^{cb} \sL_K D_cf D_a \sL_l k_b = O(f),
\ee
since terms hit by $\sL_l$ are not counted. At a Killing horizon, we 
have $T = O(f^r k^s) \Longrightarrow T=O(u^{r+s})$ where $O(u^n)$
is a function such that $\lim_{u\to 0} u^{-n+\delta} O(u^n)=0$ for any $\delta>0$.

\begin{widetext}
\section{Contractions of $Q_{abcd}, C_{ab}$}
\label{app:contr}

For the analysis of the quartic propagation cone in our EFT \eqref{lagrangian} we need the following two contractions \eqref{contr1}, \eqref{contr2}, where we define 
\be 
\epsilon^{ab}=\epsilon^{AB}(\partial_A)^a(\partial_B)^b,
\ee
so that 
\be
\epsilon^{abcd} = 6K^{[a} l^b \epsilon^{cd]},
\ee
in terms of the 4-dimensional volume element.
\begin{align}
\label{contr1}
    &(2C_{a_1a_2}C_{b_1b_2}-C_{a_1b_1}C_{a_2b_2})(\ud u)_{e_1}(\ud u)_{e_2} (\ud u)_{f_1}(\ud u)_{f_2} \epsilon^{a_1e_1g_1g_2}\epsilon^{e_2a_2h_1h_2}R_{g_1g_2h_1h_2} \epsilon^{b_1f_1i_1i_2} \epsilon^{f_2b_2j_1j_2}R_{i_1i_2j_1j_2} \notag \\
    &= \left(K^aK^bC_{ab}\right)^2 \left( \epsilon^{cd}\epsilon^{ef} R_{cdef} \right)^2 \notag \\
    &+ 2K^aK^bK^cC_{da}C_{bc} \epsilon^{de}\epsilon^{fg}K^hR_{hefg} \epsilon^{ij} \epsilon^{kl}R_{ijkl} + 2K^aK^bK^cC_{da}C_{bc} \epsilon^{ef} \epsilon^{dg} K^hR_{efhg} \epsilon^{ij} \epsilon^{kl}R_{ijkl} \notag \\
    &+ 2K^aK^bK^cC_{da}C_{bc} \epsilon^{ef} \epsilon^{gh}R_{efgh} \epsilon^{di} \epsilon^{jk}K^lR_{lijk} + 2K^aK^bK^cC_{da}C_{bc} \epsilon^{ef} \epsilon^{gh}R_{efgh} \epsilon^{ij} \epsilon^{dk} K^lR_{ijlk} \notag \\
    &+ 4K^aK^b (2C_{cd}C_{ab}-C_{ca}C_{db}) \epsilon^{ce}\epsilon^{df}K^gK^hR_{gehf} \epsilon^{ij} \epsilon^{kl} R_{ijkl} + 4K^aK^b (2C_{ca}C_{db}-C_{cd}C_{ab}) \epsilon^{ce}\epsilon^{fg}K^hR_{hefg} \epsilon^{di} \epsilon^{jk} K^l R_{lijk} \notag \\
    &+ 4 K^aK^b C_{ca}C_{db} \epsilon^{ce}\epsilon^{fg} K^h R_{hefg} \epsilon^{ij} \epsilon^{dk} K^l R_{ijlk} + 4 K^aK^b (2C_{ca}C_{db}-C_{ab}C_{cd}) \epsilon^{ef} \epsilon^{cg} K^hR_{efhg} \epsilon^{ij} \epsilon^{dk} K^l R_{ijlk} \notag \\
    &+ 4K^aK^b C_{ca}C_{db} \epsilon^{ef} \epsilon^{cg} K^h R_{efhg} \epsilon^{di} \epsilon^{jk} K^l R_{lijk} + 4 K^aK^b(2C_{ab}C_{cd} - C_{ca}C_{bd}) \epsilon^{ef} \epsilon^{gh} R_{efgh} \epsilon^{ci} \epsilon^{dj} K^kK^l R_{kilj} \notag \\
    &+ 8 K^a (2C_{ba}C_{cd}-C_{ca}C_{bd}) \epsilon^{ef} \epsilon^{bg} K^h R_{efhg} \epsilon^{ci} \epsilon^{dj} K^k K^l R_{kilj} + 8 K^a(2C_{ba}C_{cd} - C_{bc}C_{da}) \epsilon^{be} \epsilon^{fg} K^h R_{hefg} \epsilon^{ci} \epsilon^{dj} K^k K^l R_{kilj} \notag \\
    &+ 8 K^a (2C_{bc}C_{da}-C_{ba}C_{cd}) \epsilon^{be} \epsilon^{cf} K^gK^h R_{gehf} \epsilon^{ij} \epsilon^{dk} K^l R_{ijlk} +8K^a (2C_{bc}C_{da} - C_{bd}C_{ac}) \epsilon^{be} \epsilon^{cf} K^g K^h R_{gehf} \epsilon^{di} \epsilon^{jk} K^l R_{lijk} \notag \\
    &+ 16 (2C_{ab}C_{cd}-C_{ac}C_{bd}) \epsilon^{ae} \epsilon^{bf} K^gK^h R_{gehf} \epsilon^{ci} \epsilon^{dj} K^kK^l R_{kilj} \, ,
\end{align}
and
\begin{align}
\label{contr2}
    &(2C_{a_1a_2}C_{b_1b_2}-C_{a_1b_1}C_{a_2b_2}) (\ud u)_{e_1} (\ud u)_{e_2} (\ud u)_{f_1} p^c{}_{f_2}  \epsilon^{a_1e_1g_1g_2} \epsilon^{e_2a_2h_1h_2} R_{g_1g_2h_1h_2} \epsilon^{b_1f_1 i_1i_2} \epsilon^{f_2b_2j_1j_2} R_{i_1i_2j_1j_2} \notag \\
    &= -2(K^aK^b C_{ab})^2 \epsilon^{de}\epsilon^{fg} R_{defg} \epsilon^{hi}\epsilon^{cj}l^kR_{hikj} \notag \\ 
    &+ 4K^aK^bK^dC_{ea}C_{bd} \epsilon^{ef} \epsilon^{gh} K^i R_{ifgh} \epsilon^{jk} \epsilon^{cl} l^m R_{jkml} + 4K^aK^bK^dC_{ea}C_{bd} \epsilon^{fg} \epsilon^{eh} K^iR_{fgih} \epsilon^{jk}\epsilon^{cl}l^mR_{jkml} \notag \\
    &+ 4K^aK^bK^dC_{ea}C_{bd} \epsilon^{fg} \epsilon^{hi} R_{fghi} \epsilon^{ej} \epsilon^{ck} K^ll^mR_{ljmk} + 4K^aK^bK^d C_{ea}C_{bd} \epsilon^{fg} \epsilon^{hi}R_{fghi} \epsilon^{jk} \epsilon^{ce} K^ll^mR_{jklm} \notag \\
    &- 8 K^aK^b (2C_{de}C_{ab} - C_{da}C_{eb}) \epsilon^{df} \epsilon^{eg}K^hK^iR_{hfig} \epsilon^{jk} \epsilon^{cl}l^m R_{jkml} \notag \\
    &-8K^aK^b (2C_{da}C_{eb}-C_{de}C_{ab}) \epsilon^{df}\epsilon^{gh}K^iR_{ifgh} \epsilon^{ej} \epsilon^{ck} K^ll^mR_{ljmk} \notag \\
    &+ 8K^aK^bC_{da}C_{eb} \epsilon^{df} \epsilon^{gh}K^iR_{ifgh} \epsilon^{jk} \epsilon^{ce} K^ll^mR_{jklm} + 8 K^aK^b(2C_{da}C_{eb}-C_{ab}C_{db}) \epsilon^{fg} \epsilon^{dh}K^iR_{fgih} \epsilon^{jk}\epsilon^{ce}K^ll^mR_{jklm} \notag \\
    &- 8 K^aK^bC_{da}C_{eb}\epsilon^{fg}\epsilon^{dh}K^iR_{fgih} \epsilon^{ej}\epsilon^{ck}K^ll^mR_{ljmk} + 8 K^aK^b(2C_{ab}C_{de}-C_{ad}C_{be}) \epsilon^{fg} \epsilon^{hi} R_{fghi} \epsilon^{dj} \epsilon^{ce}K^kK^ll^mR_{kjlm} \notag \\
    &+ 16 K^a(2C_{bd}C_{ea}-C_{be}C_{da}) \epsilon^{bf} \epsilon^{dg} K^hK^iR_{hfig} \epsilon^{ej} \epsilon^{ck}K^ll^mR_{ljmk} \notag \\
    &- 16K^a(2C_{bd}C_{ea}-C_{be}C_{da}) \epsilon^{bf} \epsilon^{dg} K^hK^iR_{jfig} \epsilon^{jk} \epsilon^{ce}K^ll^mR_{jklm} \notag \\
    &- 16K^a(2C_{ba}C_{de}-C_{bd}C_{ae}) \epsilon^{bf} \epsilon^{gh} K^i R_{ifgh} \epsilon^{dj} \epsilon^{ce} K^kK^ll^mR_{kjlm} \notag \\
    &+ 16K^a(2C_{ba}C_{de}-C_{bd}C_{ae}) \epsilon^{fg} \epsilon^{bh} K^iR_{fgih} \epsilon^{dj} \epsilon^{ce}K^kK^ll^mR_{kjlm} \notag \\
\end{align}
\begin{align}
    &+ 32(C_{ab}C_{de}-C_{ad}C_{be}) \epsilon^{af} \epsilon^{bg}K^hK^i R_{hfig} \epsilon^{dj} \epsilon^{ce} K^kK^ll^m R_{kjlm} \notag \\
    &- 2K^aK^bK^dl^eC_{ab}C_{de} \epsilon^{fg} \epsilon^{hi}R_{fghi} \epsilon^{jk} \epsilon^{cl} K^mR_{jkml} +4K^aK^bl^d(2C_{ea}C_{bd}-C_{eb}C_{ad}) \epsilon^{fg} \epsilon^{eh}K^iR_{fgih} \epsilon^{jk} \epsilon^{cl} K^m R_{jkml} \notag \\
    &-4 K^aK^bl^dC_{ea}C_{bd} \epsilon^{ef} \epsilon^{gh} K^iR_{ifgh} \epsilon^{jk} \epsilon^{cl}K^mR_{jkml} -4K^aK^bl^d(2C_{ab}C_{ed}-C_{ae}C_{bd}) \epsilon^{fg} \epsilon^{hi}R_{fghi} \epsilon^{ej} \epsilon^{ck} K^lK^m R_{ljmk} \notag \\
    &+8K^al^b (2C_{de}C_{ab}-C_{da}C_{eb}) \epsilon^{df} \epsilon^{eg}K^hK^iR_{hfig} \epsilon^{jk}\epsilon^{cl} K^m R_{jkml} \notag \\
    &-8 K^al^b(2C_{da}C_{eb}-C_{de}C_{ab}) \epsilon^{df} \epsilon^{gh} K^iR_{ifgh} \epsilon^{ej} \epsilon^{ck}K^lK^m R_{ljmk} \notag \\
    &+ 8 K^al^b(2C_{da}C_{eb}-C_{de}C_{ab}) \epsilon^{fg} \epsilon^{dh} K^iR_{fgih} \epsilon^{ej} \epsilon^{ck}K^lK^mR_{ljmk} \notag \\
    &+ 16l^a(C_{bd}C_{ea}-C_{be}C_{da}) \epsilon^{bf} \epsilon^{dg}K^hK^iR_{hfig} \epsilon^{ej} \epsilon^{ck}K^lK^mR_{ljmk} \notag \\
    &=O(k^2,kf, f^2) \, .
\end{align}
We also need to compute certain contractions with $(C^{-1})^{ab}$. For this, we make use of the well-known formula for the inverse in terms of the adjugate of $C^{-1}$, i.e.,
\begin{align}
    &\det(Cg^{-1})(C^{-1})^{ab} = \text{adj}(C)^{ab} = \frac{1}{3!} \epsilon^{aa_1a_2a_3} \epsilon^{bb_1b_2b_3} C_{a_1b_1}C_{a_2b_2}C_{a_3b_3} \, .
\end{align}

Contractions that we need are:

\begin{align}
    &\det(Cg^{-1})(\ud u)_a(\ud u)_b (C^{-1})^{ab} \notag \\
    &= \frac{1}{3!} \epsilon^{ab}\epsilon^{cd} K^eK^f\bigg(3C_{ef}C_{ac}C_{bd}- C_{ae}C_{fc}C_{bd} + C_{ae}C_{bc}C_{fd} - C_{ec}C_{Af}C_{bd} + C_{ec}C_{ad}C_{bf} - C_{ac}C_{ed}C_{bf} - C_{ac}C_{be}C_{fd} \bigg),
\end{align}
and
\begin{align}
    &\det(Cg^{-1})p^a{}_b(\ud u)_c (C^{-1})^{bc} \notag \\
    &= \frac{1}{3!} \epsilon^{ab} \epsilon^{cd}K^eK^fl^g \bigg( C_{ge}C_{fc}C_{bd} - C_{ef}C_{gc}C_{bd} + C_{ef}C_{bc}C_{gd} - C_{be}C_{fc}C_{gd} + C_{be}C_{gc}C_{fd} - C_{ge}C_{bc}C_{fd} - C_{gc}C_{ef}C_{bd} \notag \\ 
    &\phantom{=}+ C_{gc}C_{ed}C_{bf} + C_{ec}C_{ge}C_{bd} - C_{ec}C_{gd}C_{bf} - C_{ec}C_{bf}C_{gd} + C_{ec}C_{bd}C_{fg} + C_{bc}C_{ef}C_{gd} - C_{bc}C_{ed}C_{gf} - C_{bc}C_{ge}C_{fd} \notag \\ 
    &\phantom{=}+ C_{bc}C_{gd}C_{ef} + C_{gc}C_{be}C_{fd} - C_{gc}C_{bd}C_{ef} \bigg) \, .
\end{align}
\end{widetext}
Both expressions can be simplified by applying the big-$O$-notation introduced earlier and by using the identity $\epsilon^a{}_c\epsilon^{cb}=q^{ab}$. 

\section{Proof of Theorem A}
\label{app:proofA}

By assumptions (iv), we may construct a gGNC
system $(u,v,x^A)$, see \eqref{gGNC1}, near $\cH$. We must show that 
$f,k_a=k_A(\ud x^A)_a$ vanish on $\cH$. The proof follows the inductive scheme outlined in sec. \ref{sec:thmAproofoutline}.
We take our gGNC system to be such 
that $K^a = (\partial_v)^a$, where $K^a = t^a + \omega\phi^a$, with $\omega$ a constant. Initially, $\omega=\Omega_{(0)}$, which 
is subsequently modified. At any stage, $K^a,\phi^a$ Lie-derive $(\Phi,g_{ab})$ by (v).

\subsection{Consequences of $\dot x^a =c \chi^a$ on $\cH$}

First, we investigate the consequences of (iii), (v). 

\begin{lemma}
    \label{lem1}
    On $\cH$ let $c=c(\ell,x)$ be as in (v). Then
    \begin{align}
        c =(\ud u)_a(\ud u)_b (\ud u)_c (\ud v)_dQ^{abcd} 
    \end{align}
    and $\sL_\phi c|_{\cH}=0=\sL_K c|_{\cH}$ and either $c=0$ or $c>0$ on all of $\cH$. On $\cH$, we have
    \begin{subequations}
    \begin{align}
        (\ud u)_a(\ud u)_b (\ud u)_c (\ud u)_dQ^{abcd} &= 0 \, , \label{Qrel1} \\
        (\ud u)_a(\ud u)_b(\ud u)_c p^e{}_dQ^{abcd}&=c(\Omega-\omega)\phi^e \, . \label{Qrel2} 
    \end{align}    
    \end{subequations}
\end{lemma}
\emph{Proof:} By construction, $(\ud u)_a$ is normal to $\cH$. As $\cH$ is assumed to be a characteristic hypersurface by (iii), $(\ud u)_a$ must be a characteristic covector, i.e., $Q(\ud u)=0$. Furthermore, by Hamilton's equations, we have
\begin{align}
\label{eq:lemd11}
    \dot{x}^a &= \left.\frac{\partial Q}{\partial \xi_a}\right|_{\xi=\ud u} =  (\ud u)_b (\ud u)_c (\ud u)_d Q^{bcda} \notag \\
    &=c\chi^a \, ,
\end{align}
using (v) in the last step.
Also note that $t^a=(\partial_v)^a-\omega\phi^a$ and therefore
\begin{align}
\label{eq:lem12}
    \chi^a=(\partial_v)^a + (\Omega-\omega)\phi^a \, .
\end{align}
By dotting \eqref{eq:lemd11} into $(\ud u)_a, (\ud v)_a, p^b{}_a$
in turn, we obtain all statements in the lemma. Note that these tensors, as well as $Q^{abcd}$, are Lie-derived by $\phi^a$, so that we get, e.g.,  
\begin{align}
    0=& \sL_\phi \left[ (\ud u)_a (\ud u)_b (\ud u)_c (\ud v)_d Q^{abcd} \right] \nonumber \\
    =& \sL_\phi c.
\end{align}
Thus $\sL_\phi c=0$ on $\cH$. Since $K^a$ Lie-derives the solution, the same argument works to show that $\sL_K c|_{\cH}=0$. \qed

Note that by  (v), either $c(\Omega-\omega)=0$ or 
$c(\Omega-\omega)\neq 0$ everywhere on $\cH$, so it has the same general property as $c$. In the following, we rename $c(\Omega-\omega)\to c$.

\subsection{Consequences of $Q=0$ on $\cH$}

The aim is now to evaluate the relations \eqref{Qrel1} and \eqref{Qrel2} in such a way that we can eventually deduce $c,f,k_a=0$ on $\cH$ by using the big-$O$ notation introduced above. 
\begin{lemma}
\label{lem2}
    Define $\beta_{ab}=\nabla_a\nabla_b\beta$. The contractions
    \begin{subequations}
    \begin{align}
        &K^aK^b \beta_{ab} \, , \label{eq:lem21} \\
        &p^a{}_b K^c \beta_{ac} \, , \label{eq:lem22} \\
        &K^ap^b{}_c K^dp^e{}_f R_{abde} \, , \label{eq:lem23}\\
        &K^ap^b{}_cp^d{}_e p^f{}_gR_{abdf} \, , \label{eq:lem24} \\
        &K^al^bK^cp^d{}_eR_{abcd} \label{eq:lem25}
    \end{align}
    \end{subequations}
    are all of order $O(k,f)$ on $\cH$.
\end{lemma}
\emph{Proof:} First note that since $K^a$ is Killing on $\cH$, $\sL_K\beta(\Phi)=\beta'(\Phi)\sL_K\Phi=0$ and $K_{ab}=\frac{1}{2}\sL_Kh_{ab}=0$, meaning that the curvature expressions in app. \ref{app:contr} simplify significantly. Thus, the statements for \eqref{eq:lem21}-\eqref{eq:lem24} follow immediately from the following formulas:
\begin{widetext}
\begin{align}
    K^aK^b\beta_{ab} &= - \frac{1}{2}(f+k^2)\sL_l f \sL_l\beta - \frac{1}{2}q^{ab}k_bD_a f \sL_l\beta  - \frac{1}{2} q^{ab} k_b\sL_l f D_a\beta 
    - \frac{1}{2} q^{ab} D_bfD_a\beta \, , \\
    p^a{}_bK^c \beta_{ac} &= - \left[ -\frac{1}{2} (f+k^2) \sL_l k_b - \frac{1}{2}D_b f- q^{ac} k_a D_{[b}k_{c]} \right] \sL_l\beta - \left[ \frac{1}{2} k_a\sL_lk_b - D_{[a}k_{b]} \right]q^{ac} D_c\beta \, , \\
    K^ap^b{}_cK^dp^e{}_fR_{abde} &=-2q^{ab}k_a \overline{K}_{c[f}\sL_lk_{b]} - 2q^{ab}k_a\overline{K}_{f[c} \sL_l k_{b]} + \frac{1}{4}(f+k^2) \sL_ll_c \sL_lk_f +\frac{1}{2}\sL_l k_{(c}D_{f)}f \notag \\
    &\phantom{=}- q^{ab}k_a\sL_l k_{(c}D_{f)}k_b + \frac{1}{2} D_cD_f f -\frac{1}{2} \overline{K}_{cf}q^{ab}k_aD_b f + \frac{1}{2} q^{ab} k_a \sL_l k_{(c|} D_bk_{|f)} \notag \\
    &\phantom{=}+ q^{ab}D_{[a}k_{c]}D_{[b}k_{f]} + \frac{1}{2}(f+k^2) \overline{K}_{cf}\sL_l f - \frac{1}{2}D_{(c}k_{f)} \sL_lf \, , \\
    K^ap^b{}_cp^d{}_e p^f{}_gR_{abdf} &= - \sL_lk_{[e}\overline{K}_{g]c}(f+k^2) + \overline{K}_{c[e}D_{g]}f + \frac{1}{2} \sL_l k_{[e}D_{g]}k_c - \frac{1}{2} D_ck_{[e}\sL_l k_{g]} \notag \\
    &\phantom{=}- D_cD_{[e}k_{g]} +\overline{K}_{c[e|}q^{ab}k_aD_{|g]} k_b - k_a \overline{K}_{c[e|}q^{ab} D_bk_{|g]} \, . \vphantom{\frac{a}{b}}
\end{align}
This is not as clear for the last contraction \eqref{eq:lem25}. We continue by considering the EoM's [by (i)]: 
\begin{align}
    0 &= K^aq^{bc}G_{ac}- \left(\frac{1}{2} + \alpha X\right)\sL_K\Phi q^{bc}\nabla_c\Phi  -\frac{1}{2}K^aq^{bc}g_{ac}(X-V+\frac{1}{2}\ell^2\alpha X) - \frac{1}{4}\ell^2 K^a\epsilon_a{}^{e_1e_2e_3}q^{bc}\epsilon_c{}^{f_1f_2f_3}R_{e_1e_2f_1f_2} \beta_{e_3f_3} \, .
\end{align}
The Einstein tensor satisfies (recall that $p^a{}_bK^b=0$):
\begin{align}
    K^aq^{bc}G_{ac} &= l^aK^cK^dq^{bf}R_{acdf} g^{gh} (\ud u)_g(\ud u)_h + l^aK^cK^dq^{bf} R_{acdf} g^{gh} (\ud u)_g (\ud v)_h  +  K^aK^cl^dq^{bf}R_{acdf}g^{gh}(\ud v)_g(\ud u)_h \notag \\ 
    &\phantom{=}+ l^aK^cp^d{}_eq^{bf}R_{acdg}g^{eh}(\ud u)_h + p^a{}_cK^dl^eq^{bf}R_{adef}g^{ch}(\ud u)_h + K^aq^{bc}R_{daec}q^{de} .
\end{align}
The third term vanishes due to the first Bianchi identity and the second term is the important one, since all other terms are $O(k,f)$. Thus, defining $\epsilon^{ab}=\epsilon^{AB}(\partial_A)^a(\partial_B)^b$, we have using the EoMs for $g_{ab}$: 
\begin{align}
    K^al^cK^dq^{be}R_{acde} &= - K^aq^{bc}G_{ac} + O(k,f) \vphantom{\frac{a}{b}} \notag \\
    &=-\frac{1}{4}\ell^2 \left[ -4\epsilon^{ac}\epsilon^{bd} K^el^f K^g p^h{}_a R_{efgh} \beta_{cd} + \epsilon^{ac} \epsilon^{bd} K^e p^f{}_a K^gp^h{}_d R_{efgh} l^i\beta_{ic} \right. \vphantom{\frac{a}{b}} \notag \\
    &\phantom{=} + (\ud u)_a (\ud v)_c p^b{}_d p^e{}_f \left( \epsilon^{aa_1a_2f} \epsilon^{db_1b_2c} + \epsilon^{aa_1a_2c} \epsilon^{db_1b_2f} \right) R_{a_1a_2b_1b_2} K^g \beta_{eg} \vphantom{\frac{a}{b}} \notag \\
    &\phantom{=} \left. - 2\epsilon^{ac} \epsilon^{bd} K^ep^f{}_d p^g{}_a p^h{}_c R_{efgh} K^i K^j \beta_{ij} + (\ud u)_a (\ud u)_c (\ud v)_d p^b{}_e \epsilon^{aa_1a_2c} \epsilon^{eb_1b_2d} R_{a_1a_2b_1b_2} K^fK^g \beta_{fg} \right] \vphantom{\frac{a}{b}} \notag \\ 
    &\phantom{=}+ O(k,f) \vphantom{\frac{a}{b}}
\end{align}
\end{widetext}
The prefactors arise from the possible permutations within the second and third index of the epsilon tensor, which do not give a sign due to the anti-symmetry of the Riemann tensor. From that we obtain
\begin{align}
    K^al^bK^cp^d{}_eR_{abcd} \left( \frac{1}{4} \ell^2 \epsilon^{ef}\epsilon^{gh} p^i{}_f p^j{}_h \beta_{ij} + q^{eg} \right) = O(k,f)
\end{align}
and since the term in $(...)$ is close to $q^{ab}$ for a weakly coupled theory, it is invertible. It follows that:
\begin{align}
    K^al^bK^cp^d{}_eR_{abcd} =O(k,f),
\end{align}
i.e., \eqref{eq:lem25} holds.
\qed \\

We now use these results on $\cH$ for the evaluation of the tensors appearing in the definition of $Q$ in the following lem. \ref{lem3}. To state this compactly, and for later, it is convenient 
to introduce an inner product $\langle\cdot,\cdot\rangle_h$ on tensor fields $T_{a_1\dots a_r}$ projected by $p^a{}_b$. The inner product uses $h_{ab}$ and its inverse $q^{ab}$ [see \eqref{hdef}, \eqref{pqdef}] and is given by 
\begin{equation}
\label{lrdef}
    \langle T,S \rangle_h := q^{a_1b_1} \cdots q^{a_rb_r} S_{a_1 \dots a_r} T_{b_1 \dots b_r}
\end{equation}
We denote by 
\begin{equation}
\label{normdef}
    \| T \|^2_h = \langle T,T \rangle_h
\end{equation}
the corresponding norm. By abuse of notation, we also write
\begin{equation}
    k^2 = \|k\|^2_h
\end{equation}
as before to shorten some formulas.
\begin{widetext}
\begin{lemma}
\label{lem3}
    We have,
    \begin{subequations}
    \begin{align}
        \det(Cg^{-1}) (\ud u)_a(\ud u)_b(C^{-1})^{ab} &=  f+k^2 + \ell^2 \left[ -\frac{1}{2} \langle k \sL_l f -Df,D\beta \rangle_h + f \text{tr}_h D^2\beta   +\frac{1}{2}f \sL_lf\sL_l \beta + O(k^2,kf,f^2) \right] \, , \label{lem31} \\
        (\ud u)_a(\ud u)_bP^{ab} &= -(f+k^2) + \ell^2 \bigg[ -(f+k^2)\alpha X + O(k^2,kf,f^2) \bigg] \label{lem32} \,
    \end{align}
    \end{subequations}
    on $\cH$, where \emph{tr}$_h D^2\beta=q^{ab}D_aD_b\beta$.
\end{lemma}

\emph{Proof:} It is well-known that
\begin{align}
    \det(Cg^{-1})(C^{-1})^{ab} =\text{adj}(C)^{ab} \, ,
    \label{adj1}
\end{align}
where
\begin{align}
    \text{adj}(C)^{ab}=\frac{1}{3!} \epsilon^{aa_1a_2a_3}\epsilon^{bb_1b_2b_3} C_{a_1b_1}C_{a_2b_2} C_{a_3b_3}
    \label{adj2}
\end{align}
is the adjugate of $C^{-1}$. By performing all contractions (see app. \ref{app:contr}) and using the results of lem. \ref{lem2}, we find:

\begin{align}
    &\det(Cg^{-1})(\ud u)_a(\ud u)_b (C^{-1})^{ab} \vphantom{\frac{a}{b}} \notag \\
    &= \frac{1}{2} \epsilon^{ab} \epsilon^{cd} (-f+\ell^2K^eK^f\beta_{ef})(h_{ac}+ \ell^2 \beta_{ac})(h_{bd} + \ell^2\beta_{bd}) + \frac{1}{3!}\epsilon^{ab} \epsilon^{cd} (-2k_ak_ch_{bd}+k_ak_dh_{bc} +k_ck_bh_{ad} - 2k_dk_bh_{ac}) \notag \\ 
    &\phantom{=}+ \ell^2O(k^2,kf,f^2) \vphantom{\frac{a}{b}} \notag \\
    &=f+k^2 + \ell^2 \left[ -K^aK^b\beta_{ab} + fq^{ab} \beta_{ab} + O(k^2,kf, f^2) \right] + \ell^4O(k,f) \vphantom{\frac{a}{b}} \notag \\
    &= f+k^2 + \ell^2 \left[ -\frac{1}{2} q^{ab}(k_a\sL_lf-D_af)D_b\beta + q^{ab}D_aD_b\beta \right.  \left.+ \frac{1}{2} f\sL_lf\sL_l\beta + O(k^2,kf,f^2) \right] + \ell^4O(k,f) \, ,
\end{align}
\end{widetext}
where we have defined $\epsilon^{ab}=\epsilon^{AB}(\partial_A)^a(\partial_B)^b$. For the second part of the lemma, we first note that
\begin{align}
    (\ud u)_a\nabla^a\Phi =(f+k^2)l^a\nabla_a\Phi - q^{ab}k_b\nabla_a\Phi =O(k,f)
\end{align}
using that $K^a$ is Killing. We immediately obtain:
\begin{align}
    &(\ud u)_a(\ud u)_bP^{ab} \notag \\
    &= - (1+\ell^2\alpha X)(\ud u)_a(\ud u)_bg^{ab} + \ell^2\alpha (\ud u)_a\nabla^a\Phi(\ud u)_b\nabla^b\Phi \notag \\
    &=-(f+k^2)(1+\ell^2\alpha X) + \ell^2O(k,f)^2 \notag \\
    &=  -(f+k^2) + \ell^2 (-(f+k^2)\alpha X + O(k^2,kf,f^2)) \, ,
\end{align}
which completes the proof. \qed

We do the same for contractions involving $p^a{}_b$:

\begin{lemma}
\label{lemma4}
    We have
    \begin{subequations}
    \begin{align}
        \det(Cg^{-1})p^a{}_b (\ud u)_c(C^{-1})^{bc} &= -q^{ab}k_b + \ell^2O(k,f) \, , \\
        p^a{}_b(\ud u)_c P^{bc} &= q^{ab}k_b + \ell^2O(k,f) \, .
    \end{align}
    \end{subequations}
    on $\cH$.
\end{lemma}
\emph{Proof:} As before, we use the formulas \eqref{adj1} and \eqref{adj2} to compute $C^{-1}$. From Appendix B we see that all terms involving $p^a{}_bl^cC_{ac}$ do not contribute to the $\ell^0$ part of $p^a{}_b(\ud u)_c\text{adj}(C)^{bc}$. Moreover, each term contains at least one factor of $K^aK^bC_{ab}$ or $p^a{}_bK^cC_{ac}$. Thus:
\begin{align}
    &\det(Cg^{-1})p^a{}_b(\ud u)_c(C^{-1})^{bc} \notag \\ 
    &= \frac{1}{3!} \epsilon^{ab} \epsilon^{cd} (3k_ch_{bd} - 3k_dh_{bc}) +\ell^2O(k,f) \notag \\ 
    &= - q^{ab}k_b + \ell^2O(k,f)
\end{align}
For the second part of the lemma, we have already seen that $(\ud u)_a\nabla^a\Phi=O(k,f)$, which yields:
\begin{align}
    &p^a{}_b (\ud u)_c P^{bc} \notag \\ 
    &= q^{ab}k_b + \ell^2(\alpha Xq^{ab}k_b + \alpha p^a{}_b\nabla^b\Phi O(k,f)) \notag \\ &= q^{ab}k_b + \ell^2 O(k,f).
\end{align}
\qed

\begin{widetext}
The second term in the definition of $Q$ is evaluated as follows:

\begin{lemma}
\label{lemma5}
    We have
    \begin{subequations}
    \begin{align}
        &(\ud u)_c(\ud u)_d (\ud u)_e (\ud u)_f(2C_{a_1a_2}C_{b_1b_2} - C_{a_1b_1}C_{a_2b_2}) C^{a_1cda_2}C^{b_1efb_2} \vphantom{\frac{a}{b}} \notag \\
        &=\frac{1}{4} \ell^4(\beta'(\Phi))^2 \left\| \frac{1}{2} \sL_kh\sL_lf + k \otimes D\sL_l f \right\|^2_h + \ell^4O(k^3,kf,f^2) + \ell^6O(k^2,kf,f^2) \, , \\
        &(\ud u)_c(\ud u)_d (\ud u)_ep^f{}_g(2C_{a_1a_2}C_{b_1b_2}-C_{a_1b_1}C_{a_2b_2}) C^{a_1(cd|a_2}C^{b_1|eg)b_2} =\ell^4O(k^2,kf,f^2) \vphantom{\frac{a}{b}}\, .
    \end{align}
    \end{subequations}
    on $\cH$, where $\|\cdot\|_h$ is the tensor norm associated to $h$.
\end{lemma}
\emph{Proof:} First note that, from lem. \ref{lem2},
\begin{align}
    &K^ap^b{}_cK^dp^e{}_fR_{abde} = - \frac{1}{4} \sL_kh_{cf} \sL_l f - \frac{1}{2} k_cD_f\sL_l f + O(k^2,f) =T_{cf} + O(k^2,f) 
\end{align}
\end{widetext}
as well as $K^aK^b\beta_{ab},p^a{}_bK^c\beta_{ac}=O(k,f)$, so that
\begin{align}
    K^aK^bC_{ab} = - f + \ell^2O(k,f) \, .
\end{align}
The full expression for the above contraction can be found in app. \ref{app:contr}, where the last term is the important one. As all other terms are at least $O(k^3,kf,f^2)+\ell^2O(k^2,kf,f^2)$, it remains to show that this term equals $16\|T\|_h^2+O(k^3,kf,f^2)+\ell^2O(k^2,kf,f^2)$. Note that for this purpose, we only have to consider the $\ell^0$ terms in $p^a{}_bp^c{}_dp^e{}_fp^g{}_h(2C_{ac}C_{eg}-C_{ae}C_{cg})$ since the above expression must be multiplied by $\ell^4$.

\begin{widetext}
\begin{align}
    &16p^{a_1}{}_{c_1}p^{a_2}{}_{c_2}p^{b_1}{}_{d_2}p^{b_2}{}_{d_2} (2C_{a_1a_2}C_{b_1b_2}-C_{a_1b_1}C_{a_2b_2})  \epsilon^{c_1e_1}\epsilon^{c_2e_2} K^{n_1}p^{g_1}{}_{e_1}K^{n_2}p^{g_2}{}_{e_2} R_{n_1g_1n_2g_2} \vphantom{\frac{a}{b}} \notag \\  
    &\times \epsilon^{d_1f_1} \epsilon^{d_2f_2}K^{m_1} p^{h_1}{}_{f_1} K^{m_2}p^{h_2}{}_{f_2}R_{m_1h_1m_2h_2} \vphantom{\frac{a}{b}} \notag \\
    &=16(2h_{c_1c_2}h_{d_1d_2}-h_{c_1d_1}h_{c_2d_2}) \epsilon^{c_1e_1}\epsilon^{c_2e_2} \epsilon^{d_1f_1} \epsilon^{d_2f_2}T_{e_1e_2}T_{f_1f_2} + O(k^3,kf,f^2) + \ell^2O(k^2,kf,f^2) \vphantom{\frac{a}{b}} \notag \\
    &=16q^{e_1f_1}q^{e_2f_2}T_{e_1e_2}T_{f_1f_2} + O(k^3,kf,f^2) + \ell^2O(k^2,kf,f^2) \vphantom{\frac{a}{b}} \notag \\
    &=16\left\| -\frac{1}{4} \sL_kh\sL_l f - \frac{1}{2} k \otimes D\sL_l f \right\|^2_h + O(k^3,kf,f^2) + \ell^2O(k^2,kf,f^2) \, .
\end{align}
\end{widetext}
The claim follows when multiplying by $\frac{1}{16}\ell^4 (\beta')^2$. For the second relation it suffices to check that each term contains at least two of the following factors:
\begin{align}
\begin{split}
    &K^aK^bC_{ab} , \;\,  p^a{}_bK^cC_{ac}, \; \, K^ap^b{}_cK^dp^e{}_f R_{abde} \, ,  \\
    &K^ap^b{}_cp^d{}_ep^f{}_gR_{abdf}, \; \,K^al^bK^cp^d{}_eR_{abcd} \, .
\end{split}
\end{align}
Again, the expression for a specific permutation  can be found in app. \ref{app:contr} and for the remaining permutations this result can be verified analogously. \qed 

We are now in the position to write \eqref{Qrel1} and \eqref{Qrel2} in a useful form: 

\begin{widetext}
\begin{lemma}
    We have
    \begin{align}
        &\left[f+k^2+\ell^2 \left( \frac{1}{2} \langle -k\sL_lf + Df, D\beta  \rangle_h +f\; \text{\emph{tr}}_hDD\beta + \frac{1}{2} f\sL_l f \sL_l\beta + O(k^2,kf,f^2) \right) + \ell^4O(k,f) \right] \notag \\
        &\times  \left[ f+k^2+ \ell^2 \left( (f+k^2) \alpha X + O(k^2,kf,f^2) \right) \vphantom{\frac{1}{2}} \right] \notag \\
        &= \frac{1}{4} \ell^4 (\beta')^2 \left\| \frac{1}{2} \sL_k h \sL_lf + k \otimes D\sL_l f \right\|_h^2 + \ell^4 O(k^3,kf,f^2) + \ell^6O(k^2,kf,f^2) 
        \label{Q1}
    \end{align}
    on $\cH$ and
    \begin{align}
        &\frac{1}{2} \left[ f+k^2 + \ell^2 \left( \frac{1}{2} \langle-k\sL_lf + Df,D\beta\rangle_h + f\, \text{\emph{tr}}_h DD\beta + \frac{1}{2} f\sL_lf \sL_l\beta +O(k^2,kf,f^2) \right) +\ell^4O(k,f)  \right]   \left[q^{ab}k_b+\ell^2O(k,f) \vphantom{\frac{1}{2}}\right] \notag \\
        &+\frac{1}{2} \left[ f+k^2+\ell^2\left( (f+k^2) \alpha X +O(k^2,kf,f^2) \right) \vphantom{\frac{1}{2}} \right]  \left[q^{ab}k_b + \ell^2O(k,f) \vphantom{\frac{1}{2}}\right] \notag \\
        &= c\phi^a + \ell^4O(k^2,kf,f^2)  \vphantom{\frac{1}{2}}
        \label{Q2}
    \end{align}
    on $\cH$.
\end{lemma}

\emph{Proof:} \eqref{Q1} is essentially the relation
\begin{align}
    &(\ud u)_a(\ud u)_b (\ud u)_c (\ud u)_dQ^{abcd} \notag \\ 
    &=(\ud u)_a(\ud u)_b (\ud u)_c (\ud u)_d \left( \det(Cg^{-1})(C^{-1})^{ab}P^{cd} +  (2C_{a_1a_2}C_{b_1b_2} - C_{a_1b_1}C_{a_2b_2}) C^{a_1aba_2} C^{b_1cdb_2} \right) = 0 \, ,
\end{align}
which clearly implies the first part of the lem. when using the expressions for $(\ud u)_a(\ud u)_b(C^{-1})^{ab}$, $(\ud u)_a(\ud u)_bP^{ab}$, etc. found in lem. \ref{lem3},\ref{lemma5} and multiplying by $(-1)$ on both sides, whereas \eqref{Q2} is
\begin{align}
    &(\ud u)_a(\ud u)_b (\ud u)_c p^d{}_eQ^{abce} \notag \\
    &= (\ud u)_a (\ud u)_b (\ud u)_c p^d{}_e  \left( \frac{1}{2} \det(Cg^{-1}) \left[ (C^{-1})^{ea} P^{bc} + (C^{-1})^{ab}P^{ec} \right] +  (2C_{a_1a_2}C_{b_1b_2} - C_{a_1b_1}C_{a_2b_2}) C^{a_1(ab|a_2} C^{b_1|ce)b_2} \right) \notag \\ 
    &=c\phi^d \, .
\end{align}
and using lem. \ref{lem3},\ref{lemma4},\ref{lemma5}. \qed
\end{widetext}

\subsection{Inductive Proof of Theorem A}

By (i), we may expand all tensor fields such as 
$\Phi, g_{ab}, h_{ab}, q^{ab}, p^a{}_b, c$ as power 
series in $\ell$ as in sec. \ref{sec:thmAproofoutline}. The coefficients in these expansions are analytic tensor fields on $\cM$.

The zeorth order $g_{(0)ab}$ and $\Phi_{(0)}$ are solutions to the EoM's of ordinary Einstein-scalar theory. In this case, $\cH$ is characteristic if and only if it is null. Thus, on $\cH$, $f_{(0)}=k_{(0)a}=0$ as well as $\sL_l f_{(0)}=2\kappa$, 
$D_a \Phi_{(0)}=0$, since $\cH$ is a non-degenerate Killing horizon for $\ell=0$ by (ii). 

Now we go inductively order by order in $\ell$ through \eqref{Q1} and \eqref{Q2} to successively redefine the Killing vector field $K^a \to K^a-\ell^j \omega_{(j)} \phi^a$ in every order $\ell^j$ such that, w.r.t. the gGNCs associated with the new $K^a$, 
$f_{(j)}=k_{(j)a}=0$ on $\cH$ in all orders of $\ell$. To do so, we substitute all series expansions of all quantities into \eqref{Q1} and \eqref{Q2} and compare the coefficients corresponding to a specific order in $\ell$ of both sides.

To begin, \eqref{Q2} in order $\ell^0$ and $\ell^1$ trivially implies $c_{(0)}=c_{(1)}=0$ on $\cH$, whereas from \eqref{Q1} in order $\ell^2$ we obtain $(f+k^2)_{(1)}=0=f_{(1)}$ on $\cH$, which we can use to infer from \eqref{Q2} at order $\ell^2$ that $c_{(2)}=0$ on $\cH$. At order $\ell^3$, \eqref{Q2} reads:
\begin{align}
    (f+k^2)_{(2)}k_{(1)}^a = c_{(3)}\phi^a \, ,
\end{align}
where here and in the following, we set 
\be
k^a := q^{ab}k_b,
\ee
contrary to our usual convention that indices are always raised with $g^{ab}$. From \eqref{Q1} in order $\ell^4$ we obtain $(f+k^2)_{(2)}$ on $\cH$ and therefore $c_{(3)}=0$ on $\cH$. 

\subsubsection{Induction Start: Base Case}
\label{sec:indstart}

Having dealt with these relatively straightforward consequences, we can now start the induction proper. As already outlined in sec. \ref{sec:statement}, it is possible to prove the vanishing of $f,k_a,c$ on the horizon inductively. The $n=0$ case corresponds to $f_{(j)} = 0$ for all $j\leq 3$, $k_{(j)a}=0$ for all $j\leq 1$, $c_{(j)}=0$ for all $j\leq 5$ after a suitable readjustment of $K^a$. As we will see, the key equations for the proof of these statements are given in the following lemma:
\begin{lemma}
\label{lem6}
    Define $\langle\cdot,\cdot\rangle=\langle \cdot,\cdot\rangle_{h_{(0)}}$ and $\|\cdot\|^2=\langle\cdot,\cdot\rangle$. Then we have
    \begin{subequations}
    \begin{align}
        &\frac{1}{4} (\beta_{(0)}')^2 \kappa^2 \left\| \sL_{k_{(1)}}h_{(0)} \right\|^2 \notag \\
        &=  \left[(f+k^2)_{(3)}+ \vphantom{\frac{1}{2}} \kappa \langle k_{(1)},D\beta_{(0)}\rangle \right](f+k^2)_{(3)} \, ,
        \label{Q11} \\
        &\left[ (f+k^2)_{(3)}+\frac{1}{2} \kappa \langle k_{(1)},D\beta_{(0)}\rangle \right]k_{(1)}^a = c_{(4)} \phi^a \, 
        \label{Q21}
    \end{align}
    \end{subequations}
    on $\cH$. 
\end{lemma}
\emph{Proof:} Consider \eqref{Q1} in order $\ell^6$ and \eqref{Q2} in order $\ell^4$. Whenever $(\beta')^2$ or $\|\cdot\|_h$ appear in order $\geq1$, the tensor norm always gives combinations of inner products between $(\frac{1}{2}\sL_kh\sL_lf+k\otimes D\partial_uf)_{(j_1)}$ and $(\frac{1}{2}\sL_kh\sL_lf+k\otimes D\sL_lf)_{(j_2)}$, where either $j_1=0,j_2=1$ or $j_1=1,j_2=0$. But $O(k,f)_{(0)}=0$, so that such terms vanish. Therefore, the right side of \eqref{Q1} reads:
\begin{align}
    &\frac{1}{4} (\beta'^2)_{(0)} \left\| \frac{1}{2} \sL_{k_{(1)}}h_{(0)} \sL_lf_{(0)}+k_{(1)} \otimes D\sL_lf_{(0)}\right\|^2 \notag \\  &+O(k^3,kf,f^2)_{(2)}+O(k^2,kf,f^2)_{(0)} \, .
\end{align}
Note that $\sL_lf_{(0)}=2\kappa$ on $\cH$ implies $D_a\sL_lf_{(0)}=0$ on $\cH$. The left side of \eqref{Q1} yields:
\begin{align}
    \left[(f+k^2)_{(3)}+\frac{1}{2} \langle k_{(1)}\sL_l f_{(0)},D\beta_{(0)}\rangle \right](f+k^2)_{(3)} \, ,
\end{align}
which implies \eqref{Q11}. By an analogous argument, the left side of \eqref{Q2} gives:
\begin{align}
    &\frac{1}{2} k_{(1)}^a \left[ (f+k^2)_{(3)} + \kappa \langle k_{(1)}, D\beta_{(0)}\rangle \right] \notag \\ 
    &+ \frac{1}{2} k_{(1)}^a (f+k^2)_{(3)} \, ,
\end{align}
proving \eqref{Q21}. \qed

\begin{lemma}
\label{lem7}
    The quantity 
    \be
    \zeta := (f+k^2)_{(3)}+\frac{1}{2} \kappa \langle k_{(1)},D\beta_{(0)}\rangle
    \ee
    vanishes on $\cH$.
\end{lemma}
\emph{Proof:} By lem. \ref{lem1} we either have $c=0$ or $c \neq 0$
on all of $\cH$, so the same applies to $c_{(4)}$ since it is the leading coefficient of $c$ in the expansion in $\ell$.

Case (a) $c_{(4)} = 0$: 
Either $\zeta=0$ or $k_{(1)}^a=0$ by \eqref{Q21} on some open set, hence everywhere on $\cH$ by analyticity and rotation symmetry. In both cases there is nothing to show as $k_{(1)}^a=0$ implies $f_{(2)}=f_{(3)}=0$ via \eqref{Q11}. 

Case (b) $c_{(4)} \neq 0$: By analyticity, if $\zeta$ does not vanish
identically on $\hat H:=H/U(1)$, it can vanish at most at an isolated set of $\hat H$. On an open set where 
$\eta:=c_{(4)}^{-1} \zeta$ does not vanish, we have
\begin{align}
    \langle k_{(1)}, D\beta_{(0)}\rangle
    = \beta'_{(0)} \eta^{-1} \sL_{\phi} \Phi_{(0)} = 0 \, 
\end{align}
and therefore $\zeta=(f+k^2)_{(3)}$ on this open set, hence everywhere
on $\cH$ by analyticity. 
On any open set where $\eta$ does not vanish, \eqref{Q11} yields, using $\langle k_{(1)},D\beta_{(0)}\rangle=0=\sL_\phi \eta$:
\begin{align}
    \frac{1}{2} c_{(4)}^{-2}(\beta'_{(0)})^2\kappa^2  
    \|\phi\|^2\| D\eta\|^2 
    = \eta^6 \, .
    \label{Qb1}
\end{align}
By analyticity, this relation thereby extends to all of $\cH$. 

For further analysis, it is useful to appeal to the uniformization theorem for the closed Riemannian manifold $(H, h_{(0)ab})$. By (iv) $H$ is analytically diffeomorphic to a 2-sphere, and by (v) $\phi^a$ generates a $U(1)$ isometry group. Hence, locally we can find conformally spherical coordinates $(\vartheta,\varphi)$ such that $\phi^a=(\partial_\varphi)^a$, and in these coordinates, 
\begin{align}
\label{gittel}
   h_{(0)AB} \ud x^A \ud x^B  = \Psi[ \ud \vartheta^2 + (\sin \vartheta)^2 \ud\varphi^2]
\end{align}
such that $\Psi=\Psi(\vartheta)>0$ is an analytic function on $H$ (see e.g., \cite{Gittel} for further details). Since, in these coordinates, $\phi^a$ is
a coordinate vector field and an isometry of our solution, 
$\eta, \Phi_{(0)}$ and $c_{(4)}$ are independent of $\varphi$ and \eqref{Qb1} gives us
\begin{align}
    \frac{1}{2} (\beta'_{(0)})^2 \kappa^2 c_{(4)}^{-2} (\sin \vartheta)^2 (\eta')^2 = \eta^6 ,
    \label{diffeq}
\end{align}
where $'=\ud/\ud \vartheta$. If $\eta$ is not identically zero, its zeros are isolated by analyticity. Consider two consecutive zeros 
$\vartheta_1, \vartheta_2$ of $\eta$ (there are at least two such 
zeros at $0,\pi$ by \eqref{diffeq}). By the mean value theorem, 
there exists a $\vartheta_2 \in (\vartheta_1, \vartheta_2)$ such that $\eta'(\vartheta_2) = 0$. By \eqref{diffeq}, $\vartheta_2$ therefore is another zero strictly between the two zeros $\vartheta_1, \vartheta_2$, which is a contradiction. Thus, $\eta$ is identically zero, as is thereby $\zeta$. \qed

\begin{lemma}
On $\cH$ we have $k_{(1)}^a = \omega_{(1)} \phi^a$ for some 
constant $\omega_{(1)}$, and $(f+k^2)_{(3)}=c_{(4)}=0$.
\label{lemD9}
\end{lemma}
\emph{Proof:} By the previous lemma, $\zeta=0$ and therefore $c_{(4)}=0$ on $\cH$ by \eqref{Q21}. 
We now substitute this into \eqref{Q11} to obtain\footnote{Here we have also cancelled $\kappa^2$ on both sides, which is legal since
$\kappa>0$ by (ii).}:
\begin{align}
    &\frac{1}{2} (\beta'_{(0)})^2  \left\| \sL_{k_{(1)}} h_{(0)} \right\|^2 = -  (\beta'_{(0)})^2  \left( \sL_{k_{(1)}} \Phi_{(0)} \right)^2 \, .
\end{align}
Since both sides have different signs, this equation can only be fulfilled if either $\beta_{(0)}' (=\beta'(\Phi_{(0)})) =0$, or $\sL_{k_{(1)}}h_{(0)ab} =\sL_{k_{(1)}}\Phi_{(0)}=0$ on $\cH$ or both.
The first case is excluded because we are assuming in (vi) that $\beta'$ is nowhere zero. Consequently, $k_{(1)}^a$ is a Killing vector field of $h_{(0)ab}$, and 
$(f+k^2)_{(3)}=0$ 
on $\cH$. 
Since $h_{(0)ab}$ does not have a second Killing field by (v), there exists a constant $\omega_{(1)}$ such that $k_{(1)}^a=\omega_{(1)}\phi^a$ on $\cH$. \qed

Consider now the flow $F_t$ on $H$ generated by the Killing field $-\ell k_{(1)}^a = -\ell \omega_{(1)}\phi^a$, i.e.
$x^A(t) = F_t^A(x^B)$ solves the ODE 
\be
\dot x^A(t) = -\ell k_{(1)}^A[x^B(t)], \quad x^A(0)=x^A.
\ee
Then we define the following diffeomorphism $\psi$ in an open neighborhood of $\cH$ where our gGNCs are well-defined:
\begin{equation}
    \psi(u,v,x^A) = (u,v,F_v^A(x^B)).
\end{equation}
This $\ell$ dependent analytic 
diffeomorphism maps $\cH$ to itself by construction.
We furthermore define $\Phi'= \psi^*\Phi, g_{ab}' = \psi^* g_{ab}, K^{\prime a} = \psi_* K^a$.

\begin{lemma}
\label{lem8}
    We have $K'^a=K^a-\omega_{(1)}\ell\phi^a$, and $g_{ab}'$
    takes Gaussian null form \eqref{gGNC1} with new $f',k_a',h_{ab}'$ that are analytic functions Lie-derived by 
    $K^a, \phi^a$ such that $f'_{(j)}=0$ for all $j\leq 3$, $k_{(j)a}'=0$ for all $j \leq 1$ 
    on $\cH$.
\end{lemma}

\emph{Proof:} The proof is by an elementary computation 
for pull-backs, using lem. \ref{lemD9} and the already known facts 
$f_{(0)} = f_{(1)} = f_{(2)}=k_{(0)a} = 0$ on $\cH$. \qed

We now rerun the entire argument with $(\Phi',g'_{ab})$ and $K^a$, or equivalently, with $(\Phi,g_{ab})$ and $K^{\prime a}$. Taking the second viewpoint, we use our new $K^{\prime a}$ to construct our gGNCs \eqref{gGNC1}, and we identify our spacetimes $\cM(\ell)$ by identifying points if their new gGNCs are the same. The tensors associated with the new gGNCs now satisfy $f_{(j)}'=0$ for all $j\leq 3$, $k_{(j)a}'=0$ for all $j \leq 1$. 

From \eqref{Q2} in order $\ell^5$, we obtain additionally $c_{(5)}=0$ on $\cH$. 
This concludes the base case. We rename $K^{\prime a}$ into $K^a$ to simplify the notation. \qed

\subsubsection{Induction Step: $n\rightarrow n+1$ and Proof of Theorem A}
\label{sec:indstep}

\begin{lemma}
\label{lem9}
    There exist constants $\omega_{(j)}, j\ge 1$
    such that the following is true:
    If $(u,v,x^A)$ denote the gGNCs associated to the Killing vector field $K^a - \sum_{j\geq 1} \omega_{(j)} \ell^j \phi^a$ with constants $\omega_{(j)}$ and corresponding quantities $f,k_a,c$, then $f_{(j)},k_{(j)a},c_{(j)}=0$ on $\cH$ for all $j$. 
\end{lemma}
\emph{Proof (of Theorem A):} A priori, our construction of the constants $\omega_{(j)}, j\ge 1$ in lem. \ref{lem9} does not guarantee that the series for $K^a$ converges. However, by assumption (iv), $K^a$, regarded as a formal power series, is by construction a Killing vector field such that $K_a = g_{ab}K^b$ is characteristic for $\cH$ in the formal sense (to all orders in $\ell$). Since $\chi^a$ is also such a vector with \emph{convergent} expansion in $\ell$, and since there can at most be one such vector, $K^a = \chi^a$ so the series is convergent a posteriori.
Since $K^a$ is furthermore null, so is $\chi^a$, 
thus $\cH$ is a Killing horizon. \qed

\medskip
\emph{Remark:} Note that this reasoning also 
gives $\omega_{(j)} = -\Omega_{(j)}$ $\forall j\ge 1$ a posteriori. 

\medskip
\emph{Proof (of lem. \ref{lem9}):} Our induction hypothesis is the following: For $n\in \mathbb N_0$ assume that
\begin{subequations}
    \begin{align}
        f_{(j)}|_{\cH} &=0 \text{\emph{ for all }} j\leq n+3 \, , \\
        k_{a,(j)}|_{\cH} &= 0 \text{\emph{ for all }} j\leq n+1 \, ,\\
        c_{(j)}|_{\cH} &=0 \text{\emph{ for all }} j \leq 2n+5 \, , 
    \end{align}
\end{subequations}
in gGNCs defined w.r.t. $K^a- \sum_{j=1}^{n+1} \omega_{(j)} \ell^j \phi^a$. We start by defining the following quantities:
\begin{widetext}
\begin{subequations}
\begin{align}
    A_{(j)} &= \left[f+k^2+\ell^2 \left( \frac{1}{2} \langle -k\sL_lf + Df, D\beta  \rangle_h  +f\; \text{\emph{tr}}_hDD\beta + \frac{1}{2} f\sL_l f \sL_l\beta   + O(k^2,kf,f^2) \right) +\ell^4 O(k,f) \right]_{(j)} \, , \\
    B_{(j)} &= \left[ f+k^2+ \ell^2 \left( (f+k^2) \alpha X + O(k^2,kf,f^2) \right) \vphantom{\frac{1}{2}} \right]_{(j)} \, , \\
    C_{(j)} &= \left[\frac{1}{4} \ell^4 (\beta')^2 \left\| \frac{1}{2} \sL_k h \sL_lf + k \otimes D\sL_lf \right\|_h^2 \right]_{(j)} \, .
\end{align}
\end{subequations}
\end{widetext}
In the base case $n=0$ which we have established in sec. \ref{sec:indstart}, we obtained \eqref{Q11} from \eqref{Q1} in order $\ell^6$ and \eqref{Q21} from \eqref{Q2} in order $\ell^4$, so one might expect to obtain the desired result for the induction step $n\rightarrow n+1$ by considering \eqref{Q1} in order $\ell^{2(n+1)+6}$ and \eqref{Q2} in order $\ell^{2(n+1)+4}$. We have
\begin{align}
    (AB)_{(2n+8)} &= A_{(0)}B_{(2n+8)}+ A_{(1)}B_{(2n+7)} + \dots \notag \\ &+A_{(n+4)}B_{(n+4)} + \dots + A_{(2n+8)}B_{(0)} \, .
\end{align}
We claim that the important term is $A_{(n+4)}B_{(n+4)}$, so we need to show that all other terms vanish. Note that these terms always contain one factor of $A_{(j)}$ or $B_{(j)}$ for $j< n+4$. Let $j\leq n+3$, then
\begin{align}
    (f+k^2)_{(j)} &= f_{(j)} + 2k^a_{(0)}k_{(j)a} \notag \\
    &+ \dots + 2k^a_{(\lfloor j/2 \rfloor)}k_{(\lceil j/2 \rceil)a} \, 
\end{align}
if $j$ is odd (otherwise exchange $2k^a_{(\lfloor j/2 \rfloor)}k_{(\lceil j/2 \rceil)a}^{}$ with $k^a_{(j/2)}k_{(j/2)a}^{}$). Noting that
\begin{align}
    \left\lfloor \frac{j}{2} \right\rfloor &\leq \frac{j}{2} \leq \frac{n}{2} + \frac{3}{2} \leq n+1 \text{ for all } n\geq 1 \, , \\
    j &\leq n+3 \, ,
\end{align}
and using the induction hypothesis, we can deduce $(f+k^2)_{(j)} =0$ for all $j\leq n+3,$ $n\geq 1$ (recalling the base case was already shown in sec. \ref{sec:indstart}). $O(f)$- terms with a prefactor of $\ell^2$ always have order $\leq n+3$ in $f$, hence they do vanish, as well as $k^2$- and $kf$ terms as already seen. In $(\langle k \sL_l f,D\beta\rangle_h)_{(j-2)}$, $k^a$ always appears in order $\leq n+1$, so this term gives also zero. The $\ell^4$ term and the terms in $B_{(j)}$ for $j\leq n+3$ vanish for the same reason. Hence we must have $(AB)_{(2n+8)}=A_{(n+4)}B_{(n+4)}$. Within $A_{(n+4)}$, the $\ell^2$ terms at least linear in $f$ are zero since they appear in order $\leq n+2<n+3$. For $(\langle k \sL_l f,D\beta\rangle_h)_{(n+2)}$, assume $\sL_lf,$ $D\beta$ or $h$ have order $\geq 1$, then $k$ has order $\leq n+2-1=n+1$ and therefore vanishes by assumption. Thus all quantities except $k$ must have order $0$:
\begin{align}
    (\langle k \sL_l f,D\beta\rangle_h)_{(n+2)} = \langle k_{(n+2)} \sL_lf_{(0)}, D\beta_{(0)} \rangle \, .
\end{align}
Moreover, we have already seen that the term $(f+k^2)_{(n+2)}$ vanishes, so that $B_{(n+4)}=(f+k^2)_{(n+4)}$. For the left side, assume that the order of $\beta'$ is $\geq 1$, then $\|\dots\|_h^2$ has order $\leq 2n+2$. Thus each term in $(\|\dots\|_h^2)_{(j)}$ with $j\leq 2n+2$ contains at least one factor $(\|\dots\|_h)_{(j')}$ with $j'\leq n+1$, so $f$ and $k$ always appear in order $\leq n+1$ and those terms vanish. Hence, the only non-zero term in $C_{(2n+8)}$ contains $\beta'$ and, as can be verified similarly, $h_{ab}$ and $f$ in order $0$ and thus $k_a$ in order $n+2$. The only non-trivial term left is $O(k^3)_{(2n+4)}$, which contains terms with at least one $k$-factor of order $\leq \frac{2}{3}n+\frac{4}{3}\leq n+1$ for all $n\geq 1$. Thus $O(k^3)_{(2n+4)}=0$ and, using again $D\sL_lf_{(0)}=0$ on the horizon, we obtain from \eqref{Q1}:
\begin{align}
    &\frac{1}{4}(\beta'_{(0)})^2 \kappa^2 \left\| \sL_{k_{(n+2)}} h_{(0)} \right\|^2 \notag \\ 
    &= (f+k^2)_{(n+4)} \left[ (f+k^2)_{(n+4)} +  \kappa \langle k_{(n+2)}, D\beta_{(0)} \rangle \right] \, .
    \label{Q1n}
\end{align}
Now consider the left side of \eqref{Q2} in order $\ell^{2(n+1)+4}$ and define $E^a=k^a+\ell^2O(k,f)$. First note that $E^a_{(j)}=0$ for all $j\leq n+1$ and $A_{(j)}=0$ for all $j\leq n+3$ from before. Therefore we have
\begin{align}
    (AE^a)_{(2n+6)} &= A_{(0)}E^a_{(2n+6)} +\dots + A_{(n+3)}E^a_{(n+3)} \notag \\ 
    &+ A_{(n+4)}E^a_{(n+2)} + A_{(n+5)}E^a_{(n+1)} \notag \\
    &+ \dots + A_{(2n+6)} E^a_{(0)} \notag \\ 
    &= A_{(n+4)} E^a_{(n+2)}
\end{align}
and, similarly, $(D^aB)_{(2n+6)} = D^a_{(n+2)}B_{(n+4)}$. Also note that $O(k^2,kf,f^2)_{(2n+2)}$ contains terms where $k$ and $f$ appear in order $\leq n+1$. By assumption, those terms vanish. Thus, \eqref{Q2} yields, using the expressions for $A_{(n+4)}$ and $B_{(n+4)}$ from before:
\begin{align}
    &\left[(f+k^2)_{(n+4)} +\frac{1}{2} \kappa \langle k_{(n+2)}, D\beta_{(0)} \rangle \right] k^a_{(n+2)} \notag \\ 
    &= c_{(2n+6)} \phi^a \, .
    \label{Q2n}
\end{align}
\eqref{Q1n} and \eqref{Q2n} have the same form as \eqref{Q11} and \eqref{Q21}, so we can apply the same procedure as before in sec. \ref{sec:indstart} while exchanging $(f+k^2)_{(3)} \leftrightarrow (f+k^2)_{(n+4)},$ $k^a_{(1)} \leftrightarrow k^a_{(n+2)}$ and $c_{(4)} \leftrightarrow c_{(2n+6)}$ as well as $\omega_{(1)}\ell \leftrightarrow \omega_{(n+2)} \ell^{n+2}$ within the $n=0$ case. From this procedure we can directly deduce $f_{(j)}=0$ for all $j\leq n+4$, $k_{(j)a}=0$ for all $j\leq n+2$ and $c_{(j)}=0$ for all $j\leq 2n+7$ on the horizon for gGNCs defined w.r.t. the Killing vector field $K^a-\sum_{j=1}^{n+2} \omega_{(j)} \ell^j \phi^a$. \qed

}

\section{$t$-$\varphi$ reflection symmetry $\iota$}
\label{app:iota}

Here we elucidate the reflection symmetry $\iota$
on $\cH$ referred to in assumption (iii') of sec. \ref{sec:BAss}.

We may introduce coordinates $(\varphi,\vartheta)$ on the cut $H$ as in (iii') such that $(\partial_\varphi)^a = \phi^a$. 
By assumption (iii'), we have $\iota^* \phi^a = -\phi^a, \iota^* h_{ab} = h_{ab}, \iota^* p^a{}_b = p^a{}_b, \iota^* k_a = -k_a, \iota^* f = f$ on $H$. In view of $\iota^* \phi^a = -\phi^a$, either $\iota(\varphi,\vartheta) = (-\varphi,\vartheta)$, or $\iota(\varphi,\vartheta) = (-\varphi,\pi-\vartheta)$. The last possibility 
 may be excluded as follows by considering the condition on the complex structure on $H$ in (iii'). 

The almost complex structure on $H\cong \mathbb{S}^2$ is $J^a{}_b = q^{ac} \epsilon_{cb}$, where $\epsilon_{ab} = 2\Psi \sin \vartheta (\ud \vartheta)_{[a} (\ud \varphi)_{b]}$ and $q^{ab} = \Psi^{-1}[(\sin \theta)^{-1}(\partial_\varphi)^a (\partial_\varphi)^b + (\partial_\vartheta)^a (\partial_\vartheta)^b]$ in our coordinates, see 
\eqref{gittel}. From these formulas, we immediately see that the case 
$\iota(\varphi,\vartheta) = (-\varphi,\pi-\vartheta)$ is excluded by \eqref{ioch} in  (iii'). 

By Lie-transporting the coordinates $(\varphi,\vartheta)$ along $K^a = (\partial_v)^a$, 
we get a coordinate system $(v,\varphi,\vartheta)$
on $\cH$, and in this coordinate system, $\iota$
is given by 
\be
\label{iotaform}
\iota(v,\varphi,\vartheta) = (-v,-\varphi,\vartheta).
\ee
Since $v$ is a coordinate associated with an asymtptocially timelike Killing vector field, 
we may thus think of $\iota$ as a "$t$-$\varphi$
reflection symmetry". Note that this symmetry only applies to the \emph{pull-back} of the metric on $\cH$, and is different from the symmetry $(t,\varphi) \to (-t,-\varphi)$ in Boyer-Lindquist coordinates in Kerr, which maps between the future and past horizons.

By exploiting \eqref{iotaform}, one may easily see \footnote{This is the only place where $\iota$ is used in this part of the proof.} 
\be
k^a \equiv q^{ab} k_b = F\phi^a, 
\ee
on $\cH$, for some $F$ compatible with all symmetries
$\phi^a, K^a, \iota$. Indeed, these statements may be demonstrated in a pedestrian manner using the coordinates $(v,\varphi,\vartheta)$: Using that $\iota^* k_a=-k_a$, we see that $k_a$ may not have a $\vartheta$-component, and using $\iota^*q^{ab}=q^{ab}$, we likewise see that $q^{ab}$ may not have a mixed $\vartheta\varphi$-component, and no component may depend on $(v,\varphi)$.

\end{document}